\documentclass[submission,copyright,creativecommons]{eptcs}
\providecommand{\event}{TARK 2017} % Name of the event you are submitting to
\usepackage{underscore}           % Only needed if you use pdflatex.

\title{Categories for Dynamic Epistemic Logic}
\author{Kohei Kishida%
\thanks{%
Kishida's research has been supported by the grants FA9550-12-1-0136 of the U.S. AFOSR and EP/N018745/1 of EPSRC\@.
An acknowledgment also goes to the anonymous referees for insightful comments and suggestions, which helped to improve the paper.
}%
\institute{Department of Computer Science\\
University of Oxford\\
Oxford, United Kingdom}
\email{kohei.kishida@cs.ox.ac.uk}
}
\def\titlerunning{Categories for Dynamic Epistemic Logic}
\def\authorrunning{K. Kishida}

\usepackage{amsmath}
\usepackage{amssymb}
\usepackage{amsthm}

\usepackage[only,llbracket,rrbracket]{stmaryrd}
\DeclareSymbolFont{symbolsA}{U}{pxsya}{m}{n}
\DeclareSymbolFont{symbolsC}{U}{pxsyc}{m}{n}
\DeclareMathSymbol{\Box}{\mathord}{symbolsA}{"03}
\DeclareMathSymbol{\blacksquare}{\mathord}{symbolsA}{"04}
\DeclareMathSymbol{\Diamond}{\mathord}{symbolsC}{"5E}
\DeclareMathSymbol{\Diamondblack}{\mathord}{symbolsC}{"5F}

\usepackage{tikz}
\usetikzlibrary{arrows,calc}

\def\ldot{\mathpunct{.}}

\newcommand{\mimp}{\Rightarrow}

\def\BlackBox{\blacksquare}
\def\BlackDiamond{\Diamondblack}
\newcommand{\nec}[1]{[#1]}
\newcommand{\pos}[1]{\langle #1 \rangle}
\renewcommand{\int}{\mathop{\mathrm{int}}}
\newcommand{\cmp}{\mathrel{\circ}}

\def\incto{\hookrightarrow}

\newcommand{\pw}{\mathcal{P}}
\newcommand{\Op}{\mathcal{O}}

\newcommand{\relto}{\mathrel{\ooalign{$\to$\crcr\hss\raisebox{0.1ex}{$\shortmid\mspace{1mu}$}\hss}}}
\newcommand{\opprel}[1]{{{#1}{}^\dagger}}
\newcommand{\Scott}[1]{{\llbracket{#1}\rrbracket}}
\newcommand{\Pre}{\mathrm{Pre}}
\newcommand{\NN}{\mathbb{N}}
\newcommand{\2}{\mathbf{2}}
\newcommand{\C}{{\mathbf{C}}}
\newcommand{\Rel}{{\mathbf{Rel}}}
\newcommand{\Sets}{{\mathbf{Sets}}}
\newcommand{\CABA}{{\mathbf{CABA}}}
\newcommand{\CABAvee}{{\CABA_{\vee}}}
\newcommand{\CABAwedge}{{\CABA_{\wedge}}}
\newcommand{\Kr}{{\mathbf{Kr}}}
\newcommand{\Krb}{{\Kr_\mathbf{B}}}
\newcommand{\CABAO}{{\mathbf{CABAO}}}
\newcommand{\CABAOc}{{\CABAO_\mathbf{C}}}
\newcommand{\Preord}{{\mathbf{Preord}}}
\newcommand{\Equiv}{{\mathbf{Equiv}}}
\newcommand{\Coalg}{{\mathbf{Coalg}}}
\def\op{{\mathrm{op}}}
\def\co{{\mathrm{co}}}
\def\coop{{\mathrm{coop}}}
\newcommand{\sys}[1]{\ensuremath{\mathbf{#1}}}

\newtheorem{theorem}{Theorem}
\theoremstyle{definition}
\newtheorem{corollary}{Corollary}
\newtheorem{definition}{Definition}
\newtheorem{fact}{Fact}

\begin{document}
\maketitle

\begin{abstract}
The primary goal of this paper is to recast the semantics of modal logic, and dynamic epistemic logic (DEL) in particular, in category-theoretic terms.
We first review the category of relations and categories of Kripke frames, with particular emphasis on the duality between relations and adjoint homomorphisms.
Using these categories, we then reformulate the semantics of DEL in a more categorical and algebraic form.
Several virtues of the new formulation will be demonstrated:
The DEL idea of updating a model into another is captured naturally by the categorical perspective---%
which emphasizes a family of objects and structural relationships among them, as opposed to a single object and structure on it.
Also, the categorical semantics of DEL can be merged straightforwardly with a standard categorical semantics for first-order logic, providing a semantics for first-order DEL\@.
\end{abstract}

\section{Introduction}\label{sec:intro}

\emph{Dynamic epistemic logic} (DEL) is a powerful tool at the core of ``logical dynamics'' \cite{ben11}, a logical approach to the dynamics of information and interaction.
Its semantics is general, flexible, and applicable to a wide range of informational processes in which rational agents update their knowledge and belief.
It is also malleable and admits a variety of extra structures---%
e.g.\ probabilities, preferences, questions, awareness.
It therefore forms a basis for logical studies of various aspects of agency in information and interaction.

The primary goal of this paper is to reformulate the standard semantics of DEL in category-theoretic terms.%
\footnote{\strut%
See \cite{awo10} for a clear and conceptual exposition of category theory.
}
One central idea of that semantics is that, to interpret DEL, we need to consider not just a single model but a family of models, in which one model is ``updated'' into another by a certain construction that models a given type of informational process.
This is, in fact, a kind of idea that is treated naturally from the perspective of category theory.
Category theory emphasizes a family of objects and structural relationships among them, as opposed to a single object and structure on it.
Moreover, it can compare structural relationships at a ``higher level'' among different categories, e.g.\ between a category and another that is obtained by adding extra structure to the former.
All this makes category theory excellent at capturing structural properties of a given family of models and constructions in a conceptually unifying fashion.
And this paper will show that the semantics of DEL is an instance of this.

\autoref{sec:Kripke} will lay out Kripke semantics for propositional classical modal logic from a categorical perspective.
Many of the concepts and facts covered in \autoref{sec:Kripke}, such as subframes or duality results, are found in standard expositions such as \cite{cha97,bla01};
yet we will put more emphasis on the categorical structure of Kripke frames and on ``higher'' duality between relations and algebra operations.
In \autoref{sec:del} we will use the categorical structure of Kripke frames to shed new, categorical light on the standard semantics of DEL\@.
We are not to propose a new semantics in this section, and the facts that will be covered are already known in literature (e.g.\ the standard exposition \cite{dit08}).
The point will instead be to use a categorical formulation and thereby to highlight structural properties in the standard semantics of DEL\@, uncovering the dual, algebraic ideas behind the semantics.
\autoref{sec:quantification} will give a demonstration of a virtue of our categorical, structural perspective, by showing how to extend DEL to the first order with a new, ``sheaf'' semantics.%
\footnote{\label{fn:constant.domain}\strut%
A first-order extension of dynamic logic was given in \cite{har79}.
The first attempt to extend DEL to the first order was \cite{koo07}, which introduced terms that referred to epistemic agents (and hence had a different format of logic than in this paper).
Both of these extensions used constant domains for interpreting quantification.
Constant domains can be seen as a (rather rigid) subcase of sheaves (and less flexible than sheaves in general);
their axiomatization requires the so-called Barcan formula and other axioms be added to the simple union of modal logic and first-order logic (see \autoref{thm:FOML.completeness}).
}
Clearly, our knowledge and belief and their update often involve quantified propositions, and therefore can be subject to ``first-order DEL''\@.
It may nevertheless appear extremely complicated to introduce gadgets for quantification to the DEL framework.
The structural approach, however, enables us to treat the DEL structure and the first-order structure as two modules to be simply combined.
This will make obvious the conceptual power of the approach.
\autoref{sec:connections} will discuss connections between our approach and some of the preceding categorical ones (such as the coalgebraic one).
Then \autoref{sec:conclusion} will conclude the paper, referring to lines of future work.

This article adopts the following convention when displaying facts and results:
Already known results are called ``Facts'', with references attached in footnotes.
Results that have not been explicitly stated before (to the best of the author's knowledge) are called ``Theorems'' or ``Corollaries'' (the latter follow from already known results immediately).%
\footnote{\strut%
We thank an anonymous reviewer for their suggestion of this convention.
}

\section{A Categorical Look at Kripke Semantics}\label{sec:Kripke}

This preliminary section lays out a categorical perspective on Kripke semantics for propositional classical modal logic.
We mostly consider a single pair of unary modal operators $\Box$ and $\Diamond$, but everything extends to a family of operator pairs (as we will see in the final paragraph of \autoref{sec:Kripke.frames}).

\subsection{The Category of Relations}\label{sec:Kripke.relations}

Let us first review basic facts about the category of binary relations.
Given sets $X$ and $Y$, we write $R : X \relto Y$ to mean that $R$ is a relation ``from $X$ to $Y$'', i.e.\ $R \subseteq X \times Y$.
Relations $R_1 : X \relto Y$ and $R_2 : Y \relto Z$, sharing the same $Y$, can be composed to form another $R_1;R_2 : X \relto Z$, by defining $w R_1;R_2 u$ iff $w R_1 v R_2 u$ for some $v \in Y$.
The composition is also written $R_2 \cmp R_1$ (note the opposite orders of writing $R_1$ and $R_2$).
The identity relation $w = v$ on $X$, written $1_X$, is the identity of this composition, meaning that $1_X;R = R = R;1_Y$ for every $R : X \relto Y$.
Then sets and binary relations form a category, $\Rel$.
This category comes with some extra structures, of which the most relevant to this article are the following:%
\footnote{\strut%
Categories with the following structures are studied e.g.\ in \cite{lam99}, where they are called ``ordered categories with involution''.
}
\begin{itemize}
\item
$\Rel$ is a ``dagger category'':
Each relation $R : X \relto Y$ has its opposite, $\opprel{R} : Y \relto X$, so that $v \opprel{R} w$ iff $w R v$.
This operation $\opprel{-}$ satisfies $\opprel{\opprel{R}} = R$ and $\opprel{(R_2 \cmp R_1)} = \opprel{R_1} \cmp \opprel{R_2}$, and extends to a self-dual functor $\opprel{-} : \Rel^\op \to \Rel$ by setting $\opprel{X} = X$ for each set $X$.%
\footnote{\strut%
$\Rel$ admits an even stronger structure of ``dagger compact (closed) category'', but this structure does not play an explicit r\^ole in this article.
See Subsection 3.4.2 of \cite{coe11} for $\Rel$ as a dagger compact category.
}
\item
$\Rel$ is ``locally posetal'':
For each pair of sets $X$ and $Y$, the set $\Rel(X, Y)$ of relations from $X$ to $Y$ is a poset ordered by $\subseteq$.
That is, relations $R_1, R_2 : X \relto Y$ satisfy the ``higher'' relation $R_1 \subseteq R_2$ if $w R_1 v$ implies $w R_2 v$.
Moreover, posets $\Rel(X, Y)$ and $\Rel(Y, Z)$ interact with each other in such a way that if $R_1 \subseteq R_2$ in $\Rel(X, Y)$ and $R_3 \subseteq R_4$ in $\Rel(Y, Z)$ then $R_3 \cmp R_1 \subseteq R_4 \cmp R_2$.
\item
The two structures then interact in such a way that the functor $\dagger$ gives order isomorphisms $\dagger : \Rel(X, Y) \to \Rel(Y, X)$;
i.e., $R_1 \subseteq R_2$ iff $\opprel{R_1} \subseteq \opprel{R_2}$.
\end{itemize}
The locally posetal structure makes $\Rel$ a higher category with objects (``$0$-cells'') $X$, arrows (``$1$-cells'') $R$ between objects, and higher arrows (``$2$-cells'') $\subseteq$ between arrows ($1$-cells).%
\footnote{\strut%
See Chapter XII of \cite{mac98} for this type of higher categories, ``bicategories''.
$\Rel$ appears in Subsection 1.5 (i) of \cite{lac10} as an example of bicategory.
A more general account of bicategories of relations is found in \cite{car87}.
A similar approach, in terms of categories called ``allegories'', is taken in Chapter 2 of \cite{fre90}, which also gives a thorough account of ideas in this subsection.
}
In addition, $\Rel$ satisfies
\begin{enumerate}
\setcounter{enumi}{\value{equation}}
\item\label{itm:modularity}
the ``law of modularity'':\
$((R_2 \cmp R_1) \cap R_3) \subseteq R_2 \cmp (R_1 \cap (\opprel{R_2} \cmp R_3))$ (i.e., if $w R_1 ; R_2 u$ and $w R_3 u$ then there is $v$ such that $w R_1 v$ and $v R_2 u$ and hence $w R_3 ; \opprel{R_2} v$).
\setcounter{equation}{\value{enumi}}
\end{enumerate}

Many properties of relations can be expressed with $\subseteq$.
E.g., $R : X \relto X$ is reflexive, i.e.\ $w = v$ implies $w R v$, iff $1_X \subseteq R$.
In particular, a relation $R : X \relto Y$ is a function iff both $1_X \subseteq \opprel{R} \cmp R$ and $R \cmp \opprel{R} \subseteq 1_Y$.
In addition, a function $f : X \to Y$ is injective iff $\opprel{f} \cmp f = 1_X$ and surjective iff $f \cmp \opprel{f} = 1_Y$.
Functions are thus a subcase of relations.
Moreover, the composition $R_2 \cmp R_1$ of relations is just the usual composition of functions when $R_1$ and $R_2$ are functions.
So the category $\Sets$ of sets and functions is a subcategory of $\Rel$.
On the other hand, $\Sets$ gives rise to $\Rel$ as follows.
A pair of functions $f : Z \to X$ and $g : Z \to Y$ from the same domain $Z$ is called ``jointly monic'' if $\langle f, g \rangle : Z \to X \times Y :: u \mapsto (f(u), g(u))$ is injective (or, equivalently, $(\opprel{f} \cmp f) \cap (\opprel{g} \cmp g) = 1_Z$).
Then a relation $R : X \relto Y$ corresponds to a jointly monic pair of functions, viz.\ the projections $r_1 : R \to X :: (w, v) \mapsto w$ and $r_2 : R \to Y :: (w, v) \mapsto v$ from the set $R \subseteq X \times Y$, so that the pair $(r_1, r_2)$ ``tabulates'' the relation $R : X \relto Y$, meaning that $R = r_2 \cmp \opprel{r_1}$.%
\footnote{\strut%
The correspondence mentioned here is not quite 1--1.
For two jointly monic pairs of functions $(r_1 : Z \to X, r_2 : Z \to Y)$ and $(r'_1 : Z' \to X, r'_2 : Z' \to Y)$, if there is a bijection $f : Z \to Z'$ such that $r_i = r'_i \cmp f$ for $i = 1, 2$, then the two pairs correspond to the same relation $R \subseteq X \times Y$.
One can of course identify such isomorphic pairs of jointly monic pairs and force the correspondence to be 1--1.
}

\subsection{Relation-Modality Biduality}\label{sec:Kripke.duality}

Kripke semantics uses binary relations to interpret unary modal operators.
A Kripke frame is a set $X$ paired with a binary relation $R : X \relto X$, and a Kripke model is a Kripke frame $(X, R)$ equipped with an assignment $\Scott{-}$ of subsets $\Scott{p} \subseteq X$ to propositional variables $p$.
In fact we extend the notation to all propositions $\varphi$, so that $w \in \Scott{\varphi} \subseteq X$ means that $\varphi$ is true at $w$.
Now, given a relation $R : X \relto Y$, define two monotone maps $\exists_R, \forall_R : \pw X \to \pw Y$ by
\begin{align*}
\exists_R(S)
& = \{\, v \in Y \mid w \in S \text{ for some } w \in X \text{ such that } w R v \,\}, \\
\forall_R(S)
& = \{\, v \in Y \mid w \in S \text{ for all } w \in X \text{ such that } w R v \,\} .
\end{align*}
Then, for a relation $R : X \relto X$ on a set $X$, $\exists_\opprel{R}, \forall_\opprel{R} : \pw X \to \pw X$ interpret the ``possibility'' operator $\Diamond$ and the ``necessity'' operator $\Box$, respectively---%
i.e.\
\begin{align}
\label{eq:Kripke.semantics.modal}
\Scott{\Diamond \varphi} & = \exists_\opprel{R} \Scott{\varphi} , &
\Scott{\Box \varphi} & = \forall_\opprel{R} \Scott{\varphi} .
\end{align}

An important property of $\exists_-$ and $\forall_-$ is that every relation $R$ gives an adjunction (or ``Galois connection'') $\exists_R \dashv \forall_\opprel{R}$, meaning that $\exists_R(S_1) \subseteq S_2$ iff $S_1 \subseteq \forall_\opprel{R}(S_2)$.
(And it also gives $\exists_\opprel{R} \dashv \forall_R$ via $\opprel{R}$.)
Therefore left adjoints $\exists_R$ preserve arbitrary joins and right adjoints $\forall_R$ preserve arbitrary meets.
It also needs noting that a relation $f : X \relto Y$ is a function iff $\exists_\opprel{f} = \forall_\opprel{f}$, in which case $\exists_\opprel{f} = \forall_\opprel{f}$ is the inverse-image map $f^{-1} : \pw Y \to \pw X$.
So, for every function $f$, the map $f^{-1} = \exists_\opprel{f} = \forall_\opprel{f}$ preserves all joins and meets, and moreover $\exists_f \dashv f^{-1} \dashv \forall_f$, which is one of the fundamental facts for categorical logic.%
\footnote{\strut%
The idea that $\Box$ and $\Diamond$ are a relational generalization of $\forall_f$ and $\exists_f$ is laid out in \cite{her11} from a more general categorical perspective of ``bicategories of relations and spans''.
As observed in \cite{her11}, we can define modal operators $\BlackDiamond$ and $\BlackBox$ that are ``opposite'' to $\Diamond$ and $\Box$, and interpret them with $\exists_R$ and $\forall_R$;
then we have adjunctions $\BlackDiamond \dashv \Box$ and $\Diamond \dashv \BlackBox$, which also appear in \cite{kur13}.
These adjunctions are typical of the ``past'' and ``future'' modalities of temporal logic, as observed in \cite{kar98}.
}

One of the most fundamental categorical facts to the interpretation \eqref{eq:Kripke.semantics.modal} is the equivalence of $\Rel$ and categories of complete atomic Boolean algebras (CABAs).
Let $\CABAvee$ and $\CABAwedge$ be the categories of CABAs with all-join-preserving maps and with all-meet-preserving maps, respectively, and then

\begin{fact}\label{thm:equivalence.Kripke}
$\exists_- :: R \mapsto \exists_R$ and $\forall_- :: R \mapsto \forall_R$ extend to equivalences of categories $\exists_- : \Rel \to \CABAvee$ and $\forall_- : \Rel \to \CABAwedge$, both sending a set $X$ to its powerset $\pw X$, while every CABA has the form $\pw X$.%
\footnote{\strut%
See, e.g., Exercise 5.2.5 in \cite{jac16} for essentially the same fact.
}
Putting this in ``concrete'' terms,
\begin{itemize}
\item
The relations $R : X \to Y$ correspond 1-1 to the all-join-preserving maps $\exists_R : \pw X \to \pw Y$, and also 1-1 to the all-meet-preserving maps $\forall_R : \pw X \to \pw Y$.
In other words, for every pair of sets $X$ and $Y$, each of $\exists_-$ and $\forall_-$ induces a bijection from $\Rel(X, Y)$ to the set $\C(\pw X, \pw Y)$ of arrows of $\C = \CABAvee, \CABAwedge$ from $\pw X$ to $\pw Y$.%
\end{itemize}
\end{fact}

In fact, higher versions of \autoref{thm:equivalence.Kripke} are relevant to modal logic.
Recall that $\Rel$ is equipped with higher arrows between arrows, i.e.\ the relations $\subseteq$ among relations $R_1, R_2 : X \relto Y$.
Similarly, $\C = \CABAvee, \CABAwedge$ are also equipped with the relation $\leqslant$ among arrows $h_1, h_2 : \pw X \to \pw Y$, by setting $h_1 \leqslant h_2$ iff $h_1(S) \subseteq h_2(S)$ for all $S \in \pw X$, making $\C(\pw X, \pw Y)$ a poset.
Then

\begin{fact}\label{thm:biduality.Kripke.1}
$\exists_- : \Rel \to \CABAvee$ is a (higher) equivalence,%
\footnote{\strut%
See Subsection 1.5 (i) of \cite{lac10} for essentially the same fact.
}
meaning that
\begin{itemize}
\item
$R_1 \subseteq R_2$ iff $\exists_{R_1} \leqslant \exists_{R_2}$,
\end{itemize}
i.e., each bijection $\exists_- : \Rel(X, Y) \to \CABAvee(\pw X, \pw Y)$ is an order isomorphism.
\end{fact}

Yet there are more versions of this result that are less frequently mentioned but equally important.
Since $\Rel$ has two levels of arrows, $R$ and $\subseteq$, there are four types of (higher) functors $F$ from $\Rel$ to another higher category $\C$, viz.,
$F : \Rel \to \C$, with the direction of neither $R$ nor $\subseteq$ flipped;
$F : \Rel^\op \to \C$, with just $R$ flipped;
$F : \Rel^\co \to \C$, with just $\subseteq$ flipped;
$F : \Rel^\coop \to \C$, with both $R$ and $\subseteq$ flipped.%
\footnote{\strut%
$\Rel^\op$, $\Rel^\co$, and $\Rel^\coop$ are $\Rel$
with just $R$ flipped;
with just $\subseteq$ flipped;
and
with both $R$ and $\subseteq$ flipped.
See Subsection 1.6 of \cite{lac10}.
}
Therefore there can be four versions of equivalence (or duality).

\begin{corollary}\label{thm:biduality.Kripke.2}
$\forall_- : \Rel^\co \to \CABAwedge$ is a ``$2$-cell duality'', i.e.\ an equivalence flipping $\subseteq$ (but not $R$).
Concretely put,
\begin{itemize}
\item
$R_1 \subseteq R_2$ iff $\forall_{R_2} \leqslant \forall_{R_1}$ (note the reversed order),
\end{itemize}
i.e., each bijection $\forall_- : \Rel(X, Y) \to \CABAwedge(\pw X, \pw Y)$ is an order-reversing isomorphism.
\end{corollary}

Moreover, composing $\exists_-$ and $\forall_-$ with the self-dual functor $\dagger : \Rel^\op \to \Rel$, which is a ``$1$-cell duality'', i.e.\ an equivalence flipping $R$ but not $\subseteq$, we obtain

\begin{corollary}\label{thm:biduality.Kripke.3}
$\exists_\opprel{-} : \Rel^\op \to \CABAvee$ is a $1$-cell duality, and $\forall_\opprel{-} : \Rel^\coop \to \CABAwedge$ is a ``biduality'', i.e.\ an equivalence flipping both $R$ and $\subseteq$.
Concretely put,
\begin{enumerate}
\setcounter{enumi}{\value{equation}}
\item\label{thm:biduality.Kripke.functoriality}
$\exists_\opprel{(R_2 \cmp R_1)} = \exists_\opprel{R_1} \cmp \exists_\opprel{R_2}$ and $\forall_\opprel{(R_2 \cmp R_1)} = \forall_\opprel{R_1} \cmp \forall_\opprel{R_2}$ (note the flipped orders of composition),
\item\label{thm:biduality.Kripke.exists}
$R_1 \subseteq R_2$ iff $\exists_\opprel{R_1} \leqslant \exists_\opprel{R_2}$,
\item\label{thm:biduality.Kripke.forall}
$R_1 \subseteq R_2$ iff $\forall_\opprel{R_2} \leqslant \forall_\opprel{R_1}$ (note the reversed order),
\setcounter{equation}{\value{enumi}}
\end{enumerate}
i.e., $\exists_\opprel{-}$ and $\forall_\opprel{-}$ induce order-preserving isomorphisms $\exists_\opprel{-} : \Rel(Y, X) \to \CABAvee(\pw X, \pw Y)$ and order-reversing isomorphisms $\forall_\opprel{-} : \Rel(Y, X) \to \CABAwedge(\pw X, \pw Y)$, respectively.
\end{corollary}

Thus, \eqref{eq:Kripke.semantics.modal} means that the modal operators $\Diamond$ and $\Box$ are duals to the relations $R : X \relto X$, in such a ``higher'' way that the relation $\subseteq$ among the latter corresponds to the relation $\leqslant$ among the former (e.g.\ the reflexivity of $R$, i.e.\ $1_X \subseteq R$, is equivalent by \eqref{thm:biduality.Kripke.exists} to $1_{\pw X} \leqslant \exists_\opprel{R}$ and by \eqref{thm:biduality.Kripke.forall} to $\forall_\opprel{R} \leqslant 1_{\pw X}$, i.e.\ $\varphi \vdash \Diamond \varphi$ and $\Box \varphi \vdash \varphi$).%
\footnote{\strut%
See \cite{kis13} for more on correspondence results via the higher dualities.
}
This higher duality plays a fundamental r\^ole in this article as well as in Kripke semantics in general.

One more fact that will prove useful is the ``Beck-Chevalley condition'':

\begin{corollary}\label{thm:beck.chevalley}
If the following diagram is a pullback in $\Sets$, then $p \cmp \opprel{q} = \opprel{f} \cmp g$.%
\footnote{\strut%
$\Sets$ satisfies the Beck-Chevalley condition, meaning that the pullback entails $\exists_q \cmp p^{-1} = g^{-1} \cmp \exists_f$.
See Section IV.9 of \cite{mac92}.
This implies \eqref{eq:beck.chevalley} by \autoref{thm:biduality.Kripke.3}.
}
\begin{gather}\label{eq:beck.chevalley}
\begin{gathered}
\begin{tikzpicture}[x=20pt,y=20pt]
\coordinate (O) at (0,0);
\coordinate (r) at (3.25,0);
\coordinate (d) at (0,-2.5);
\node (A0) [inner sep=0.25em] at (O) {$Y \times_X Z$};
\node (A1) [inner sep=0.25em] at ($ (A0) + (r) $) {$Z$};
\node (B0) [inner sep=0.25em] at ($ (A0) + (d) $) {$Y$};
\node (B1) [inner sep=0.25em] at ($ (B0) + (r) $) {$X$};
\draw [->] (A0) -- (A1) node [pos=0.5,inner sep=2pt,above] {$q$};
\draw [->] (B0) -- (B1) node [pos=0.5,inner sep=2pt,below] {$f$};
\draw [->] (A0) -- (B0) node [pos=0.5,inner sep=2pt,left] {$p$};
\draw [->] (A1) -- (B1) node [pos=0.5,inner sep=2pt,right] {$g$};
\coordinate (A0-pb) at ($ (A0) + (1,-1) $);
\draw ($ (A0-pb) + (-0.45,0) $) -- (A0-pb) -- ($ (A0-pb) + (0,0.45) $);
\end{tikzpicture}
\end{gathered}
\end{gather}
\end{corollary}

\subsection{Categories of Kripke Frames}\label{sec:Kripke.frames}

Let us now consider categories of Kripke frames.
A \emph{monotone} map from a Kripke frame $(X, R_X)$ to another $(Y, R_Y)$ is a function $f : X \to Y$ that preserves relation, i.e.\ such that $w R_X v$ implies $f(w) R_Y f(v)$.
Observe that this can equivalently be written as either of the following.
\begin{enumerate}
\setcounter{enumi}{\value{equation}}
\item\label{itm:continuous.Kripke.2}
$R_X \subseteq \opprel{f} \cmp R_Y \cmp f$ (i.e., $w R_X v$ implies $w f w' R_Y v' \opprel{f} v$ for some $w', v' \in Y$),
\item\label{itm:continuous.Kripke.3}
$f \cmp R_X \subseteq R_Y \cmp f$ (i.e., $w R_X v f v'$ implies $w f w' R_Y v'$ for some $w' \in Y$).
\setcounter{equation}{\value{enumi}}
\end{enumerate}
The formulation \eqref{itm:continuous.Kripke.3} strengthens to $f$ being a \emph{bounded morphism}, i.e.\ satisfying both \eqref{itm:continuous.Kripke.3} and
\begin{enumerate}
\setcounter{enumi}{\value{equation}}
\item\label{itm:open.Kripke.1}
$R_Y \cmp f \subseteq f \cmp R_X$ (i.e., $w f w' R_Y v'$ implies $w R_X v f v'$ for some $v \in X$),
\setcounter{equation}{\value{enumi}}
\end{enumerate}
i.e.\ satisfying
\begin{enumerate}
\setcounter{enumi}{\value{equation}}
\item\label{itm:open.Kripke.2}
$f \cmp R_X = R_Y \cmp f$.
\setcounter{equation}{\value{enumi}}
\end{enumerate}
Let us write $\Kr$ for the category of Kripke frames and monotone maps, and $\Krb$ for its subcategory of bounded morphisms.

The duality observed in \autoref{sec:Kripke.duality} immediately entails duality results between Kripke frames and ``CABAs with operators'' (CABAOs), i.e.\ CABAs equipped with all-join-preserving operators $\Diamond$.
The isomorphisms $\exists_\opprel{-} : \Rel(X, X) \to \CABAvee(\pw X, \pw X)$ in \autoref{thm:biduality.Kripke.3} mean that the Kripke frames $(X, R)$ correspond 1-1 to the CABAOs $(\pw X, \Diamond)$.
Moreover, while the functions $f : X \to Y$ and the CABA homomorphisms $h : \pw Y \to \pw X$ are dual to each other, \autoref{thm:biduality.Kripke.3} further implies (by $f^{-1} = \exists_\opprel{f} = \forall_\opprel{f}$) that \eqref{itm:open.Kripke.2} is equivalent to
\begin{enumerate}
\setcounter{enumi}{\value{equation}}
\item\label{itm:open.Kripke.algebra}
$\exists_\opprel{R_X} \cmp f^{-1} = f^{-1} \cmp \exists_\opprel{R_Y}$ (or equivalently $\forall_\opprel{R_X} \cmp f^{-1} = f^{-1} \cmp \forall_\opprel{R_Y}$),
\setcounter{equation}{\value{enumi}}
\end{enumerate}
i.e., $f^{-1}$ being a CABAO homomorphism, i.e.\ a CABA homomorphism that moreover preserves $\Diamond$ (and $\Box$), from $(\pw Y, \exists_\opprel{R_Y})$ to $(\pw X, \exists_\opprel{R_X})$.
Therefore the category $\Krb$ is dual to the category $\CABAO$ of CABAOs and CABAO homomorphisms.%
\footnote{\strut%
This duality was first shown in \cite{tho75}.
See also \cite{bla01}.
}
In fact, let us call a CABA homomorphism $h$ ``continuous'' if it has $\Diamond \cmp h \leqslant h \cmp \Diamond$, and then \autoref{thm:biduality.Kripke.3} implies that \eqref{itm:continuous.Kripke.2}--\eqref{itm:continuous.Kripke.3} are equivalent to
\begin{enumerate}
\setcounter{enumi}{\value{equation}}
\item\label{itm:continuous.Kripke.algebra}
$\exists_\opprel{R_X} \leqslant f^{-1} \cmp \exists_\opprel{R_Y} \cmp \exists_f$, or equivalently $\exists_\opprel{R_X} \cmp f^{-1} \leqslant f^{-1} \cmp \exists_\opprel{R_Y}$ (or $f^{-1} \cmp \forall_\opprel{R_Y} \cmp \forall_f \leqslant \forall_\opprel{R_X}$ or $f^{-1} \cmp \forall_\opprel{R_Y} \leqslant \forall_\opprel{R_X} \cmp f^{-1}$),
\setcounter{equation}{\value{enumi}}
\end{enumerate}
i.e.\ the continuity of $f^{-1}$.
Hence the category $\Kr$ is dual to the category $\CABAOc$ of CABAOs and continuous CABA homomorphisms \cite{ghi10}.
We should stress, however, that these duality results are merely derivative, and that the dualities in \autoref{sec:Kripke.duality} are more fundamental.
It is the latter duality that we will take essential advantage of throughout this article.

We have so far considered a single pair of operators $\Box$ and $\Diamond$, but in epistemic logic we often take a set $A$ of agents and consider a pair of operators $\nec{\alpha}$ (also written $K_\alpha$, for ``$\alpha$ knows that'') and $\pos{\alpha}$ for each agent $\alpha \in A$.
To interpret this $A$-indexed set of operator pairs, a Kripke frame $X$ needs to be equipped with an $A$-indexed set of relations $R_\alpha : X \relto X$ as well.
Let us say that a function $f : X \to Y$ from a Kripke frame $(X, R^\alpha_X)_{\alpha \in A}$ to another $(Y, R^\alpha_Y)_{\alpha \in A}$ is monotone if it preserves every $R^\alpha_X$ by satisfying \eqref{itm:continuous.Kripke.2}--\eqref{itm:continuous.Kripke.3} (with $R^\alpha_X$ in place of $R_X$), and a bounded morphism if it satisfies \eqref{itm:continuous.Kripke.2}--\eqref{itm:open.Kripke.2} for every $R^\alpha_X$ (in place of $R_X$).
Then the Kripke frames with $A$-many relations and their monotone maps or bounded morphisms form categories $\Kr_A$ and $\Krb_A$, subsuming $\Kr$ and $\Krb$ above as just a special case with $A$ a singleton.
The duality results in this section carry over straightforwardly to $\Kr_A$ and $\Krb_A$, with respect to CABAs with $A$-many operators.

\subsection{Topological Constructions for Kripke Frames}\label{sec:Kripke.topological}

Having introduced two categories of Kripke frames, it may appear to be a natural question which of the two we should adopt as ``the'' category of Kripke frames.
The answer is, however, that we need both $\Kr$ and $\Krb$.
The significance of $\Krb$ is fairly obvious and well studied.
Bounded morphisms are dual to homomorphisms preserving $\Diamond$ and $\Box$ as well as all the other connectives, and therefore closely connected to the preservation of modal logic.
Indeed, the bisimulations are precisely the ``relations in $\Krb$'' (see the final paragraph of \autoref{sec:del.pal}).
By the same token, in the coalgebraic approach to Kripke semantics, the kind of homomorphisms considered are those corresponding to bounded morphisms, and hence the considered category of coalgebras is equivalent to $\Krb$ (see \autoref{sec:connections} for more on the connection to the coalgebraic approach).
Quite arguably, $\Krb$ plays a more prominent r\^ole than $\Kr$ does, as long as the ``static'' modal logic is concerned.
Nevertheless, this statement no longer applies to the semantics of dynamic epistemic logic (DEL).
Many of the semantic constructions crucial for DEL take place in $\Kr$ but not in $\Krb$.
Indeed, to let DEL show interesting behaviors, it is essential to use monotone maps and not bounded morphisms.

The category $\Kr$ admits a wide range of constructions that are directly connected to ones in $\Sets$ using sets and functions.
They are due to

\begin{fact}\label{thm:Kr.topological}
$\Kr$ is ``topological over $\Sets$'',%
\footnote{\strut%
See Section 21 of \cite{ada90} for the definition and nice properties of topological categories.
(It may need noting that \cite{ada90} refers to $\Kr$ as $\Rel$.)
This subsection refers to Definitions 21.1 and 21.7, Example 21.8, Propositions 21.30 and 21.31, Theorem 21.9, and Proposition 21.15.
}
meaning, concretely, the following.
Given any family of functions $f_i : X \to Y_i$ ($i \in I$) to Kripke frames $(Y_i, R_i)$, the relation
\begin{gather*}
w R_X v \iff f_i(w) R_i f_i(v) \text{ for all } i \in I ,
\quad\text{i.e.}\quad
R_X = \bigcap_{i \in I} (\opprel{f_i} \cmp R_i \cmp f_i) ,
\end{gather*}
is the (unique) ``initial lift'' of $\{ f_i \}_{i \in I}$, i.e.\ the relation on $X$ such that, given any function $g : Z \to X$, all $f_i \cmp g$ are monotone from a frame $(Z, R_Z)$ iff $g$ is.
\end{fact}

(In fact, \autoref{thm:Kr.topological} holds of $\Kr_A$ in general, again with $R^\alpha_X$ in place of $R_X$.)
One may note that the relation $R_X$ in \autoref{thm:Kr.topological} is the largest relation on $X$ preserved by all $f_i$, since, for every relation $R$ on $X$,
\begin{gather}
\label{eq:initial.continuous}
R \subseteq R_X
\iff R \subseteq \opprel{f_i} \cmp R_i \cmp f_i \text{ (i.e.\ } f_i \text{ preserves } R \text{) for all } i \in I .
\end{gather}
It is easy to observe that initial lifts preserve many properties of relations such as reflexivity, transitivity, and symmetry.
Then the full subcategories of $\Kr$ given by those properties and combinations thereof, such as $\Preord$ of the preorders (i.e.\ reflexive and transitive relations) and $\Equiv$ of the equivalence relations, are said to be ``initially closed''.
It follows that these subcategories are also topological over $\Sets$, and that the inclusion functors have left-adjoints.%
\footnote{\strut%
There are properties that are not preserved by initial lifts.
E.g., antisymmetry is not;
in fact, the category of posets is not topological over $\Sets$.
}
E.g., the left adjoint $F : \Kr \to \Preord$ sends a Kripke frame $(X, R)$ to $(X, R^\ast)$, where $R^\ast$ is the reflexive and transitive closure of $R$.

One consequence of $\Kr$, or a subcategory such as $\Preord$, being topological over $\Sets$ is that it also has ``final lifts'', dual to initial lifts of \autoref{thm:Kr.topological}.
E.g., given a family of preorders $(X, R_\alpha)$ ($\alpha \in A$) on the same set $X$, such as ``epistemic'' relations $R_\alpha$ of agents $\alpha \in A$, consider an $A$-indexed family of identity maps $\{ 1_X \}_{\alpha \in A}$ in $\Sets$;
then its final lift in $\Preord$ comes with the epistemic relation for the ``common knowledge'' of the group $A$, i.e.\ $(\bigcup_\alpha R_\alpha)^\ast$.%
\footnote{\strut%
See Section 2.3 of \cite{dit08}, as well as \cite{bal04,bal98}, for common knowledge.
We do not treat its logic in this article.
}

Another consequence, more relevant to this article, is that the forgetful functor $U : \Kr \to \Sets$ to the complete and cocomplete category $\Sets$ lifts limits and colimits---%
meaning that, given any (small) diagram $D$ in $\Kr$, its (co)limit exists on the (co)limit of $U \cmp D$ in $\Sets$.
Most notably,
\begin{enumerate}
\setcounter{enumi}{\value{equation}}
\item\label{itm:Kripke.limit.product}
Given a family of Kripke frames $(Y_i, R_i)$ ($ i \in I$), its product in $\Kr$, $(X, R_X)$, is defined on the cartesian product $X = \prod_{i \in I} Y_i$ by taking $R_X = \bigcap_{i \in I} (\opprel{p_i} \cmp R_i \cmp p_i)$ for the projections $p_i : X \to Y_i$.
\item\label{itm:Kripke.limit.equalizer}
Let $i : S \incto X$ be an inclusion map.
Then $(S, R_S)$ is a subframe of a Kripke frame $(X, R_X)$, i.e.\ $R_S = \opprel{i} \cmp R_X \cmp i$, iff $i$ is a regular mono from $(S, R_S)$ to $(X, R_X)$ in $\Kr$.
\setcounter{equation}{\value{enumi}}
\end{enumerate}
These constructions, and their canonical maps $p_i$ and $i$, are crucial to the semantics of DEL, as we will see in \autoref{sec:del}.
Pullbacks in $\Kr$ will also play a key r\^ole later in \autoref{sec:quantification.fodel}.
In particular, observe

\begin{theorem}\label{thm:pullback.preserve.open}
The pullback of a bounded morphism in $\Kr$ is a bounded morphism.%
\footnote{\strut%
This is a straightforward analogue of the already known fact that, in the category of topological spaces, the pullback of an open map is open.
See Proposition 1 in Section V.4 of \cite{joy84}.
}
\end{theorem}

\begin{proof}
Let \eqref{eq:beck.chevalley} be a pullback in $\Kr$, let $R_X$ be the relation on $X$, similarly for $Y$, $Z$, and $Y \times_X Z$, and let $g$ be a bounded morphism.
Then $p$ satisfies \eqref{itm:open.Kripke.1} as follows, by
\eqref{itm:continuous.Kripke.2} for $f$;
the commuting of \eqref{eq:beck.chevalley};
\eqref{itm:open.Kripke.2} for $g$;
\autoref{thm:beck.chevalley};
the law of modularity \eqref{itm:modularity};
and
the definition of $R_{Y \times_X Z}$ as the initial lift of $p$ and $q$.
\begin{align*}
R_Y \cmp p
  \subseteq (R_Y \cmp p) \cap (\opprel{f} \cmp R_X \cmp f \cmp p)
& = (R_Y \cmp p) \cap (\opprel{f} \cmp R_X \cmp g \cmp q) \\
& = (R_Y \cmp p) \cap (\opprel{f} \cmp g \cmp R_Z \cmp q) \\
& = (R_Y \cmp p) \cap (p \cmp \opprel{q} \cmp R_Z \cmp q) \\
& \subseteq p \cmp ((\opprel{p} \cmp R_Y \cmp p) \cap (\opprel{q} \cmp R_Z \cmp q))
  = p \cmp R_{Y \times_X Z} .
\qedhere
\end{align*}
\end{proof}

It needs stressing, however, that the canonical maps of ``topological'' constructions in this subsection are not in general bounded morphisms, and hence do not live in $\Krb$.
Indeed, as we will see, they must not be bounded morphisms for DEL to show interesting behaviors.

\section{A Categorical Look at Dynamic Epistemic Logic}\label{sec:del}

This section shows how to use the categorical structure of \autoref{sec:Kripke} to reformulate the standard semantics of dynamic epistemic logic (DEL) structurally.
We will first review the simpler subcase of \emph{public announcement logic} (PAL) in \autoref{sec:del.pal}, and then expand it to the general DEL in \autoref{sec:del.del}.

\subsection{Public Announcement Logic}\label{sec:del.pal}

Regular monos $i$ of $\Kr$ in \eqref{itm:Kripke.limit.equalizer} are used to interpret PAL\@.
This logic has unary operators $\nec{\sigma !}$ and $\pos{\sigma !}$ for all of its propositions $\sigma$.
The proposition $\nec{\sigma !} \varphi$ is intended to mean ``$\varphi$ will be the case after $\sigma$ is publicly and truthfully announced (or observed)'', and interpreted as follows:
Given a Kripke model $(X, R_X, \Scott{-}_X)$ and a subset $S = \Scott{\sigma}_X$ with inclusion $i : S \incto X$, let $(S, R_S, \Scott{-}_S)$ be the submodel on $S$---%
which is defined by $R_S = \opprel{i} \cmp R_X \cmp i$ and $\Scott{p}_S = i^{-1} \Scott{p}_X$ for atomic $p$.
Then
\begin{enumerate}
\setcounter{enumi}{\value{equation}}
\item\label{itm:PAL.semantics.box}
$w \in \Scott{\nec{\sigma !} \varphi}_X$ iff either $w \notin \Scott{\sigma}_X$ or $w \in \Scott{\varphi}_S$ (note the subscripts), i.e., iff $v \in \Scott{\varphi}_S$ for all $v \in S$ such that $v i w$.
In short, $\Scott{\nec{\sigma !} \varphi}_X = \forall_i \Scott{\varphi}_S$.
\setcounter{equation}{\value{enumi}}
\end{enumerate}
Similarly (or De Morgan-dually),
\begin{enumerate}
\setcounter{enumi}{\value{equation}}
\item\label{itm:PAL.semantics.diamond}
$w \in \Scott{\pos{\sigma !} \varphi}_X$ iff both $w \in \Scott{\sigma}_X$ and $w \in \Scott{\varphi}_S$, i.e., iff $v \in \Scott{\varphi}_S$ for some $v \in S$ such that $v i w$.
In short, $\Scott{\pos{\sigma !} \varphi}_X = \exists_i \Scott{\varphi}_S$.
\setcounter{equation}{\value{enumi}}
\end{enumerate}

One may contrast \eqref{itm:PAL.semantics.box} and \eqref{itm:PAL.semantics.diamond} to
\begin{align}
\label{itm:PAL.semantics.pre}
\forall_i \cmp i^{-1} \Scott{\varphi}_X & = \Scott{\sigma \mimp \varphi}_X , &
\exists_i \cmp i^{-1} \Scott{\varphi}_X & = \Scott{\sigma \wedge \varphi}_X .
\end{align}
So, although generally $\Scott{\varphi}_S \neq i^{-1} \Scott{\varphi}_X$, for atomic $p$ we have $\Scott{p}_S = i^{-1} \Scott{p}_X$ by definition, and hence have a ``reduction axiom'' $\nec{\sigma !} p \equiv (\sigma \mimp p)$ by
\begin{gather*}
\Scott{\nec{\sigma !} p}_X = \forall_i \Scott{p}_S = \forall_i \cmp i^{-1} \Scott{p}_X = \Scott{\sigma \mimp p}_X .
\end{gather*}
Reduction axioms, taken together for atomic sentences and for all the ``static'' connectives, completely axiomatize PAL by reducing it to the static modal logic.
Proofs for reduction axioms for connectives are:
\begin{enumerate}
\setcounter{enumi}{\value{equation}}
\item\label{itm:PAL.reduction.conj}
Because $\forall_i$ preserves meets,
\begin{align*}
\Scott{\nec{\sigma !}(\varphi \wedge \psi)}_X
& = \forall_i(\Scott{\varphi}_S \cap \Scott{\psi}_S)
  = \forall_i \Scott{\varphi}_S \cap \forall_i \Scott{\psi}_S
  = \Scott{\nec{\sigma !} \varphi \wedge \nec{\sigma !} \psi}_X .
\end{align*}
\item\label{itm:PAL.reduction.neg}
A CABA homomorphism, $i^{-1}$ preserves $\lnot$.
And $\opprel{i} \cmp i = 1_S$, or dually $i^{-1} \cmp \forall_i = 1_{\pw(S)}$, since $i$ is an injection.
Therefore $\lnot_S = \lnot_S \cmp i^{-1} \cmp \forall_i = i^{-1} \cmp \lnot_X \cmp \forall_i$.
Hence
\begin{align*}
\Scott{\nec{\sigma !} \lnot \varphi}_X
& = \forall_i \cmp \lnot_S \Scott{\varphi}_S
  = \forall_i \cmp i^{-1} \cmp \lnot_X \cmp \forall_i \Scott{\varphi}_S
  = \Scott{\sigma \mimp \lnot \nec{\sigma !} \varphi}_X .
\end{align*}
\item\label{itm:PAL.reduction.box}
$R_S = \opprel{i} \cmp R_X \cmp i$ dually means $\forall_\opprel{R_S} = i^{-1} \cmp \forall_\opprel{R_X} \cmp \forall_i$.
Therefore
\begin{align*}
\Scott{\nec{\sigma !} \Box \varphi}_X
& = \forall_i \cmp \forall_\opprel{R_S} \Scott{\varphi}_S
  = \forall_i \cmp i^{-1} \cmp \forall_\opprel{R_X} \cmp \forall_i \Scott{\varphi}_S
  = \Scott{\sigma \mimp \Box \nec{\sigma !} \varphi}_X .
\end{align*}
\setcounter{equation}{\value{enumi}}
\end{enumerate}
These algebraic proofs are straightforward applications of properties of the duality $\forall_\opprel{-}$.
In particular, it should be noted that \eqref{itm:PAL.reduction.box}, the reduction via $\Box$, is simply a dual to the equality of relations $R_S \cmp \opprel{i} = \opprel{i} \cmp R_X \cmp i \cmp \opprel{i}$.

A perspective on \eqref{itm:PAL.semantics.box}--\eqref{itm:PAL.semantics.diamond} that has been guiding the study of the semantics of PAL, and indeed of DEL (see e.g.\ \cite{bal04}), is that $\nec{\sigma !}$ and $\pos{\sigma !}$ are interpreted by $\forall_i$ and $\exists_i$, and therefore are the modal operators of the relation $\opprel{i}$ (called a ``transition relation'' in \cite{bal04}), similarly to $\Box$ and $\Diamond$ interpreted by $\forall_\opprel{R}$ and $\exists_\opprel{R}$ of $R$ as in \eqref{eq:Kripke.semantics.modal}.
One difference is that, whereas $R$ is a relation on the same set, $\opprel{i}$ is between different sets.
Thus PAL, and DEL in general, generalize Kripke semantics by using relations $R : X \relto Y$ between different Kripke frames to interprete modal operators.
In studying this general setting, it proves helpful to use the relation-modality dualities of \autoref{sec:Kripke.duality} (and not just the derivative dualities of \autoref{sec:Kripke.frames} between Kripke frames and CABAOs).
It may also be interesting to note that $\sigma \mimp {-}$ and $\sigma \wedge {-}$ in \eqref{itm:PAL.semantics.pre}, which play an essential r\^ole in reduction axioms, are modal operators, too, viz.\ those of the relation $i \cmp \opprel{i} : X \relto X$.
This is the reason the relation-modality duality $\forall_\opprel{-}$ is applicable in \eqref{itm:PAL.reduction.box}.

A point of caution here for our categorical approach is that, in general, $R : X \relto Y$ is neither a structure on a Kripke frame (an object of the category $\Kr$) nor a monotone map (an arrow of $\Kr$).
So, to accommodate it in terms of $\Kr$, we use the idea of tabulation from \autoref{sec:Kripke.relations}:
A relation $R : X \relto Y$ corresponds to the pair of projections $r_1 : R \to X$ and $r_2 : R \to Y$ from the set $R \subseteq X \times Y$, so that $R = r_2 \cmp \opprel{r_1}$.
Indeed, given Kripke frames on $X$ and $Y$, \autoref{thm:Kr.topological} gives a canonical Kripke frame on $R \subseteq X \times Y$ from which $r_1$ and $r_2$ are monotone.
Then $\forall_\opprel{R} = \forall_{r_1} \cmp {r_2}^{-1}$ and $\exists_\opprel{R} = \exists_{r_1} \cmp {r_2}^{-1}$;
hence $\forall_\opprel{R}$ and $\exists_\opprel{R}$ of all relations $R$ can be obtained by $\exists_f \dashv f^{-1} \dashv \forall_f$ of monotone maps $f$.
This trick, using monotone maps $r_1$ and $r_2$ of $\Kr$, always works for any relation $R \subseteq X \times Y$.
On the other hand, bounded morphisms of $\Krb$ do not always work, since $r_1$ and $r_2$ are both bounded morphisms if and only if $R$ is a bisimulation.
(We will see an even more crucial r\^ole of $\Kr$ at the end of \autoref{sec:del.del}.)

\subsection{Dynamic Epistemic Logic}\label{sec:del.del}

Let us now consider the Baltag-Moss-Solecki semantics of DEL \cite{bal98} and observe how product update in it can be treated categorically.
Take two Kripke frames, $(X, R_X)$ and $(E, R_E)$, and regard the former as an ``epistemic model'' and the latter as an ``event model''.
So, let us assume that $(X, R_X)$ is equipped with an interpretation $\Scott{\Pre(e)}_X \subseteq X$ of the precondition $\Pre(e)$ of every event $e \in E$ (or we can take a Kripke model $(X, R_X, \Scott{-}_X)$ on $(X, R_X)$);
we write $i_e : \Scott{\Pre(e)}_X \incto X$ for the inclusion maps.
Then the product model of the two frames, obtained by ``updating'' $(X, R_X)$ with $(E, R_E)$, is defined on the disjoint union of $i_e$, i.e.\ the subset
\begin{gather}
\label{eq:del.product.update}
X \otimes E = \sum_{e \in E} \Scott{\Pre(e)}_X = \{\, (w, e) \in X \times E \mid w \in \Scott{\Pre(e)}_X \,\}
\quad\text{of}\quad
X \times E = \sum_{e \in E} X .
\end{gather}
The ``epistemic'' relation $R_{X \otimes E}$ on $X \otimes E$ is defined as the subframe of the product $(X \times E, R_{X \times E})$ of $(X, R_X)$ and $(E, R_E)$, using \eqref{itm:Kripke.limit.product} and \eqref{itm:Kripke.limit.equalizer}.
This amounts to
\begin{gather*}
(w_1, e_1) R_{X \otimes E} (w_2, e_2) \iff w_1 R_X w_2 \text{ and } e_1 R_E e_2 ,
\quad\text{i.e.,}\quad
R_{X \otimes E} = (\opprel{p_X} \cmp R_X \cmp p_X) \cap (\opprel{p_E} \cmp R_E \cmp p_E)
\end{gather*}
for the projections $p_X : X \otimes E \to X :: (w, e) \mapsto w$ and $p_E : X \otimes E \to E :: (w, e) \mapsto e$.
In short, it is the initial lift of $p_X$ and $p_E$.
In addition, given a Kripke model $\Scott{-}_X$ on $X$, it induces an updated Kripke model on $X \otimes E$ by $\Scott{p}_{X \otimes E} = {p_X}^{-1} \Scott{p}_X$ for atomic $p$.%
\footnote{\strut%
This is the case without ``factual change''.
A version with factual change \cite{ben06} can also be treated categorically.
}

Let us analyze this construction a bit further, using the following diagram (for each $e \in E$).
\begin{gather}
\label{eq:del.product.update.diagram}
\begin{gathered}
\begin{tikzpicture}[x=20pt,y=20pt]
\coordinate (O) at (0,0);
\coordinate (r1) at (5.75,0);
\coordinate (r2) at (3.75,0);
\coordinate (d) at (0,-2.5);
\node (A0) [inner sep=0.25em] at (O) {$\sum_{e' \in E} \Scott{\Pre(e')}_X$};
\node (A0L) [anchor=east,inner sep=0.25em] at ($ (A0.west) + (0.5em,0) $) {$X \otimes E = {}$};
\node (A1) [inner sep=0.25em] at ($ (A0) + (r1) $) {$\sum_{e' \in E} X = X \times E$};
\node (A2) [inner sep=0.25em] at ($ (A1) + (r2) $) {$E$};
\node (B0) [inner sep=0.25em] at ($ (A0) + (d) $) {$\Scott{\Pre(e)}_X$};
\node (B1) [inner sep=0.25em] at ($ (B0) + (r1) $) {$X$};
\draw [right hook->] (A0) -- (A1) node [pos=0.5,inner sep=2pt,above] {$i$};
\draw [->] (A1) -- (A2) node [pos=0.5,inner sep=2pt,above] {$p'_E$};
\draw [right hook->] (B0) -- (B1) node [pos=0.5,inner sep=2pt,below] {$i_e$};
\draw [>->] (B0) -- (A0) node [pos=0.5,inner sep=2pt,left] {$q_e$};
\draw [transform canvas={xshift=-5pt},>->] (B1) -- (A1) node [pos=0.5,inner sep=2pt,left] {$q'_e$};
\draw [transform canvas={xshift=5pt},->] (A1) -- (B1) node [pos=0.5,inner sep=2pt,right] {$p'_X$};
\coordinate (B0-pb) at ($ (B0) + (1,1) $);
\draw ($ (B0-pb) + (-0.45,0) $) -- (B0-pb) -- ($ (B0-pb) + (0,-0.45) $);
\end{tikzpicture}
\end{gathered}
\end{gather}
Here $p'_X$ and $p'_E$ are the obvious projections, so that $p_X = p'_X \cmp i$ and $p_E = p'_E \cmp i$.
And $q_e$ and $q'_e$ are the ``coproduct injections'' $w \mapsto (w, e)$.
The inclusion $i : X \otimes E \incto X \times E$ has $i \cmp q_e = q'_e \cmp i_e$ (by its definition as $i = \sum_{e' \in E} i_{e'}$), while $p'_X \cmp q'_e = 1_X$ (since $p'_X$ equals the trivial ``cotuple'' $[1_X]_{e' \in E}$), and therefore $p_X \cmp q_e = p'_X \cmp i \cmp q_e = p'_X \cmp q'_e \cmp i_e = i_e$.

Given this construction, for each $e \in E$ the canonical functions $i_e$ and $q_e$ tabulate a relation $R_e = q_e \cmp \opprel{i_e} : X \relto X \otimes E$;
i.e., $w R_e (w', e')$ iff $w = w' \in \Scott{\Pre(e)}_X$ and $e = e'$, or $w R_e v$ iff $v p_X w$ and $v p_E e$.
\label{page:beck.chevalley.product.update}%
\autoref{thm:beck.chevalley} implies $R_e = q_e \cmp \opprel{i_e} = \opprel{i} \cmp q'_e$ since the square in \eqref{eq:del.product.update.diagram} is a pullback.
This relation, and its duals $\forall_\opprel{R_e} = \forall_{i_e} \cmp {q_e}^{-1}$ and $\exists_\opprel{R_e} = \exists_{i_e} \cmp {q_e}^{-1}$, are then used to interpret the dynamic operators $\nec{E, e}$ and $\pos{E, e}$;
the proposition $\nec{E, e} \varphi$ is supposed to mean ``$\varphi$ will be the case after the event $e$ takes place''.
The interpretation, similar to \eqref{itm:PAL.semantics.box}--\eqref{itm:PAL.semantics.diamond}, is as follows:
\begin{align}
\label{itm:DEL.semantics.modal}
\Scott{\nec{E, e} \varphi}_X
& = \forall_\opprel{R_e} \Scott{\varphi}_{X \otimes E} , &
\Scott{\pos{E, e} \varphi}_X
& = \exists_\opprel{R_e} \Scott{\varphi}_{X \otimes E} .
\end{align}

As in \eqref{itm:PAL.semantics.pre}, relations $p_X \cmp R_e = p_X \cmp q_e \cmp \opprel{i_e} = i_e \cmp \opprel{i_e}$ give
\begin{align}
\label{itm:DEL.semantics.pre}
\forall_{i_e} \cmp {i_e}^{-1} \Scott{\varphi}_X & = \Scott{\Pre(e) \mimp \varphi}_X , &
\exists_{i_e} \cmp {i_e}^{-1} \Scott{\varphi}_X & = \Scott{\Pre(e) \wedge \varphi}_X ,
\end{align}
which we may call ``static precondition modalities'', as the modal operators of $i_e \cmp \opprel{i_e}$.
Then the reduction axioms of DEL can be proven as follows.
(The reduction via $\wedge$ goes since $\forall_\opprel{R_e}$ preserves meets, just the same way as in \eqref{itm:PAL.reduction.conj};
the case of $\lnot$ is similar to \eqref{itm:PAL.reduction.neg}, albeit more complicated.)
\begin{enumerate}
\setcounter{enumi}{\value{equation}}
\item\label{itm:DEL.reduction.atom}
$p_X \cmp R_e = i_e \cmp \opprel{i_e}$ implies the following for atomic $p$, by \eqref{itm:DEL.semantics.pre} and $\Scott{p}_{X \otimes E} = {p_X}^{-1} \Scott{p}_X$.
\begin{align*}
\Scott{\nec{E, e} p}_X
& = \forall_\opprel{R_e} \Scott{p}_{X \otimes E}
  = \forall_\opprel{R_e} \cmp {p_X}^{-1} \Scott{p}_X
  = \forall_{i_e} \cmp {i_e}^{-1} \Scott{p}_X
  = \Scott{\Pre(e) \mimp p}_X .
\end{align*}
\item\label{itm:DEL.reduction.box}
For the case of $\Box$, first note that $w R_e v p_E e'$ implies $e = e'$ since $w R_e v$ implies $v p_E e$ whereas $p_E$ is a function.
In other words, $w R_e v p_E e'$ iff $w R_e v$ and $e' = e$.
This entails $(\ast)$ in the following:
\begin{align*}
w R_e ; R_{X \otimes E} v
& \iff w R_e v' p_X ; R_X ; \opprel{p_X} v \text{ and } w R_e v' p_E ; R_E ; \opprel{p_E} v \text{ for some } v ' \in X \otimes E \\
& \stackrel{(\ast)}{\iff} w R_e ; p_X ; R_X ; \opprel{p_X} v \text{ and } e R_E ; \opprel{p_E} v \\
& \stackrel{(\dagger)}{\iff} w R_e ; p_X ; R_X ; R_{e'} v \text{ for some } e' \in E \text{ such that } e R_E e' ,
\end{align*}
where $(\dagger)$ holds since $u R_{e'} v$ iff $v p_X u$ and $v p_E e'$, i.e.\ iff $u \opprel{p_X} v$ and $e' \opprel{p_E} v$.
Thus,
\begin{gather*}
R_{X \otimes E} \cmp R_e = (\bigcup_{e R_E e'} R_{e'}) \cmp R_X \cmp p_X \cmp R_e = (\bigcup_{e R_E e'} R_{e'}) \cmp R_X \cmp i_e \cmp \opprel{i_e} .
\end{gather*}
Observe on the other hand that, for a family of relations $R_i : X \relto Y$ of the same type, we have $\forall_\opprel{(\bigcup_i R_i)} = \bigcap_i \cmp \forall_\opprel{R_i}$.
Therefore
\begin{align*}
\Scott{\nec{E, e} \Box \varphi}_X
& = \forall_\opprel{R_e} \cmp \forall_\opprel{R_{X \otimes E}} \Scott{\varphi}_{X \otimes E} \\
& = \forall_{i_e} \cmp {i_e}^{-1} \cmp \forall_\opprel{R_X} \cmp \bigcap_{e R_E e'} \cmp \forall_\opprel{R_{e'}} \Scott{\varphi}_{X \otimes E}
  = \Scott{\Pre(e) \mimp \Box \bigwedge_{e R_E e'} [E, e'] \varphi} {\vphantom\varphi}_X .
\end{align*}
\setcounter{equation}{\value{enumi}}
\end{enumerate}

We conclude this section with a remark on the significance of using the category $\Kr$.
We reviewed in this section that topological constructions (\autoref{sec:Kripke.topological}) and their canonical maps play essential r\^oles in the semantics of PAL and DEL\@.
These constructions take place in $\Kr$ as opposed to the category $\Krb$, and the canonical maps are monotone maps of $\Kr$, and not bounded morphisms of $\Krb$.
Indeed, for DEL to show interesting behaviors, the canonical maps---%
in particular, $p_X : X \otimes E \to X$, which amounts to $i : \Scott{\sigma}_X \incto X$ in the case of PAL---%
must not be bounded moprphisms.
For, if $p_X$ is a bounded morphism, then $\Scott{\varphi}_{X \otimes E} = {p_X}^{-1} \Scott{\varphi}_X$ for every $\varphi$ and not just atomic $p$ (this entails $\nec{E, e} \varphi \equiv \Pre(e) \mimp \varphi$ the same way as in \eqref{itm:DEL.reduction.atom})---%
this means that no event can teach agents anything.
In other words, for events to teach agents something, they must bring about some change logically, and therefore the maps $f$ representing them must not have logic-preserving duals $f^{-1}$.

\section{Application:\ Quantification}\label{sec:quantification}

This section demonstrates a virtue of our categorical perspective, by showing how to extend DEL to the first order.
Our structural approach to DEL and the standard structural approach to first-order logic can be integrated together, simply as two modules, using the methodology of category theory.
We will first review how to interpet classical first-order logic in \autoref{sec:quantification.classical}, and how to add this first-order structure to Kripke semantics using ``Kripke sheaves'' in \autoref{sec:quantification.kripke}.
We will then equip this semantics with a DEL-type update in \autoref{sec:quantification.fodel}, obtaining a new sheaf semantics for first-order DEL\@.

\subsection{Classical Semantics in a Slice Category}\label{sec:quantification.classical}

Here we review how the standard semantics for classical first-order logic goes in the category $\Sets / X$, as the non-modal basis of semantics in \autoref{sec:quantification.kripke}.
See \cite{pit00} for a more general and extensive account.

Let us first recall the definition of \emph{slice category}.
Given any category $\C$, fix any object $C$.
Then the slice category $\C / C$, ``$\C$ over $C$'', consists of the following:
\begin{itemize}
\item
Objects are any arrow $f : D \to C$ of $\C$ with the codomain $C$.
\item
Arrows from $f : D \to C$ to $g : E \to C$ are any arrow $h : D \to E$ of $C$ such that $g \cmp h = f$.
\end{itemize}
In particular, given a set $X$, $\Sets / X$ is the category of ``sets and functions over $X$'':
\begin{itemize}
\item
Objects, ``sets over $X$'', are functions $\pi : D \to X$.
For each $w \in X$ we write $D_w$ for the inverse image $\pi^{-1}(\{ w \})$, called the ``fiber over $w$''.
\item
And arrows from $\pi_1 : D \to X$ to $\pi_2 : E \to X$ are functions $f : D \to E$ ``over $X$'', meaning that $\pi_2 \cmp f = \pi_1$, or equivalently that if $a \in D_w$ then $f(a) \in E_w$ for the same $w$.
\end{itemize}
We will also later consider a Kripke-structured version of $\Sets / X$, viz.\ $\Kr / (X, R)$ over a Kripke frame $(X, R)$:
Its objects and arrows are monotone maps and not just any functions.

Fixing any (nonempty) set $X$, the slice category $\Sets / X$ is used to interpret classical first-order logic as follows.
We fix an object $\pi : D \to X$ of $\Sets / X$, and a surjection $\pi$ in particular.
We then regard $X$ as a set of worlds and $D$ as a set of individuals.
Each individual $a \in D$ is assumed to live in a unique world, viz.\ $\pi(a) \in X$.
In this sense we may call $\pi$ a ``residence map''.
For each world $w \in X$, the fiber $D_w = \pi(\{ w \})$ is the set of individuals living in $w$.
In fact, for each $n \in \NN$, the cartesian product $D^n_w = D_w \times \cdots \times D_w$ is the set of $n$-tuples of individuals living in $w$, and the disjoint union of $D^n_w$ for all $w \in X$, i.e.\ the $n$-fold ``fibered product'' of $D$ over $X$,
\begin{gather*}
D^n_X = \sum_{w \in X} D^n_w = \{\, (a_1, \ldots, a_n) \in D \times \cdots \times D \mid \pi(a_1) = \cdots = \pi(a_n) \,\} ,
\end{gather*}
is the set of $n$-tuples from the same world, with the projection
\begin{gather*}
\pi^n : D^n_X \to X :: \bar{a} \to \pi(a_i)
\end{gather*}
mapping an $n$-tuple from the same world to that world.
(As special cases, $D^1_X = D$ and $D^0_X = X$, with $\pi^1 = \pi : D \to X$ and $\pi^0 = 1_X : D \to D$.)
Categorically speaking, this is to take the $n$-fold pullback of $D$ over $X$ in $\Sets$, or the $n$-fold product of $\pi$ in $\Sets / X$.

One important note regarding the semantics in $\Sets / X$ is that it interprets ``formulas in contexts''.
A context is a (finite) sequence of variables that are all distinct.
A formula $\varphi$ can be in a context $(x_1, \ldots, x_n)$ if no other variables occur freely in $\varphi$.
It is not assumed that all of $x_1$, \ldots, $x_n$ actually occur freely in $\varphi$;
so, e.g., if $\varphi$ can be in a context $(x_1, \ldots, x_n)$ then it can also be in $(x_1, \ldots, x_n, x_{n + 1}, \ldots x_m)$.
A formula-in-context is a pair of formula and a context it can be in;
so, writing $(\, x_1, \ldots, x_n \mid \varphi \,)$ presupposes that $\varphi$ can be in $(x_1, \ldots, x_n)$.
Now, we semantically interpret formulas-in-contexts $(\, x_1, \ldots, x_n \mid \varphi \,)$ rather than formulas $\varphi$:
We regard $(\, x_1, \ldots, x_n \mid \varphi \,)$ as an $n$-ary predicate that may or may not be true of $n$-tuples of individuals $(a_1, \ldots, a_n)$.
Similarly, we interpret terms-in-contexts $(\, x_1, \ldots, x_n \mid t \,)$ as mappings of $n$-tuples of individuals to individuals.
We will write $\bar{x}$ and $\bar{a}$ for sequences $(x_1, \ldots, x_n)$ and $(a_1, \ldots, a_n)$.

In propositional logic, we interpret a sentence $\sigma$ with $\Scott{\sigma} \subseteq X$, so that $w \in \Scott{\sigma}$ means that $\sigma$ is true at $w$.
Similarly, in the semantics in $\Sets / X$, we interpret a closed sentence $\sigma$ in the empty context with $\Scott{\sigma} \subseteq X$.
Yet, extending this, we interpret an $n$-ary formula-in-context $(\, \bar{x} \mid \varphi \,)$ with $\Scott{\, \bar{x} \mid \varphi \,} \subseteq D^n_X$, so that $\bar{a} \in \Scott{\, \bar{x} \mid \varphi \,}$ means that $\varphi$ is true of individuals $a_1$, \ldots, $a_n$ in place of $x_1$, \ldots, $x_n$ (at the world $\pi(a_i)$).
The same formula $\varphi$ in different contexts is true of different tuples:
E.g.\ $(a, b) \in \Scott{\, x, y \mid \varphi \,}$ iff $(b, a) \in \Scott{\, y, x \mid \varphi \,}$ iff $(a, b, c) \in \Scott{\, x, y, z \mid \varphi \,}$ (for any $c \in D$ such that $(a, b, c) \in D^3_X$).

An interpretation $\Scott{-}$ can be defined inductively, first for terms and then for formulas.
In interpreting terms in $\Sets / X$, the core idea is to interpret an $n$-ary term-in-context $(\, \bar{x} \mid t \,)$ with an arrow $\Scott{\, \bar{x} \mid t \,} : D^n_X \to D$ in $\Sets / X$, i.e.\ a function sending $\bar{a} \in D^n_w$ to $\Scott{\, \bar{x} \mid t \,}(\bar{a}) \in D_w$.
To each $n$-ary function symbol $f$, assign an arrow $\Scott{f} : D^n_X \to D$ of $\Sets / X$.
(This includes $\Scott{c} : X \to D$ for a constant, i.e.\ $0$-ary function symbol.)
Then, for the base case let $\Scott{\, \bar{x} \mid f \bar{x} \,} = \Scott{f}$, whereas we also let $\Scott{\, \bar{x} \mid x_i \,} = p_i : D^n_X \to D :: \bar{a} \mapsto a_i$ for each $i \leqslant n$.
For inductive steps, define the substitution of terms as follows:
Given a term-in-context $(\, x_1, \ldots, x_n \mid t \,)$ and terms $t_1, \ldots, t_n$, we write $t [t_1, \ldots, t_n / x_1, \ldots, x_n]$ for the result of substituting $t_i$ for all the free occurrences of $x_i$ in $t$.
Then, given $\Scott{\, \bar{x} \mid t \,}$ and $\Scott{\, \bar{y} \mid t_i \,}$ for each $i \leqslant n$ where $\bar{y} = (y_1, \ldots, y_m)$, write
\begin{gather*}
\Scott{\, \bar{y} \mid \bar{t} \,} = \langle \Scott{\, \bar{y} \mid t_1 \,}, \ldots, \Scott{\, \bar{y} \mid t_n \,} \rangle : D^m_X \to D^n_X :: \bar{b} \mapsto (\Scott{\, \bar{y} \mid t_1 \,}(\bar{b}) , \ldots, \Scott{\, \bar{y} \mid t_n \,}(\bar{b})) ,
\end{gather*}
and we have
\begin{gather}
\label{eq:FOML.semantics.term}
\Scott{\, \bar{y} \mid t [t_1, \ldots, t_n / x_1, \ldots, x_n] \,} = \Scott{\, \bar{x} \mid t \,} \cmp \Scott{\, \bar{y} \mid \bar{t} \,} .
\end{gather}

Now, to each $n$-ary relation symbol $F$, assign any subset $\Scott{F} \subseteq D^n_X$, and $\Scott{\, \bar{x} \mid F \bar{x} \,} = \Scott{F}$.
Inductively,
\begin{align}
\label{eq:FOML.semantics.boolean}
\Scott{\, \bar{x} \mid \varphi \wedge \psi \,} & = \Scott{\, \bar{x} \mid \varphi \,} \cap \Scott{\, \bar{x} \mid \psi \,} , &
\Scott{\, \bar{x} \mid \lnot \varphi \,} & = \lnot \Scott{\, \bar{x} \mid \varphi \,} = D^n_X \setminus \Scott{\, \bar{x} \mid \varphi \,}
\end{align}
for Boolean operators.
For quantifiers, take a projection $p : D^{n + 1}_X \to D^n_X :: (\bar{a}, b) \mapsto \bar{a}$ and let
\begin{align}
\label{eq:FOML.semantics.quantify}
\Scott{\, \bar{x} \mid \forall y \ldot \varphi \,} & = \forall_p \Scott{\, \bar{x}, y \mid \varphi \,} , &
\Scott{\, \bar{x} \mid \exists y \ldot \varphi \,} & = \exists_p \Scott{\, \bar{x}, y \mid \varphi \,} ;
\end{align}
the case of $n = 0$ is just $p = \pi : D \to X$.
Closely connected to quantification is the substitution of terms:
Write $\varphi [\bar{t} / \bar{x}]$ for the result of substituting $t_i$ for $x_i$ in $\varphi$ (this makes sense only if $t$ is free for $x$ in $\varphi$).
Then the substitution satisfies
\begin{gather}
\label{eq:FOML.semantics.substitution}
\Scott{\, \bar{y} \mid \varphi [t_1, \ldots, t_n / x_1, \ldots, x_n] \,} = \Scott{\, \bar{y} \mid \bar{t} \,}^{-1} \Scott{\, \bar{x} \mid \varphi \,} .
\end{gather}
As an instance of this, given $\Scott{\, \bar{x} \mid \varphi \,}$ we can add a vacuous variable to the context by
\begin{gather}
\label{eq:FOML.semantics.vacuous}
\Scott{\, \bar{x}, y \mid \varphi \,} = p^{-1} \Scott{\, \bar{x} \mid \varphi \,}
\end{gather}
for the same $p :: (\bar{a}, b) \mapsto \bar{a}$ as above;
and other operations on contexts (e.g.\ permutation) can be interpreted in similarly obvious ways.

\subsection{Kripke-Sheaf Semantics}\label{sec:quantification.kripke}

In this subsection we review ``Kripke-sheaf semantics'' for first-order modal logic.
An extensive exposition of this semantics is in \cite{gab09}.
We use the notation and terminology from \cite{kis11}, however, to be consistent with \autoref{sec:quantification.classical}.%
\footnote{\strut%
\cite{kis11} provides a more general semantics using neighborhood structure, but Kripke-sheaf semantics is simply a special case involving Kripke frames;
see Section 3 of \cite{kis11}, in particular.
It should be noted that the definitions of Kripke sheaf in \cite{gab09} (Definition 3.6.2) and in \cite{kis11} (Definition 3.5) only agree for the limited case of reflexive and transitive Kripke frames.
\autoref{def:Kripke.sheaf} in the following is the version in \cite{kis11}.
}

As to syntax, we take a first-order language---%
with relation symbols, variables, function symbols and constants---%
and add $\Box$ and $\Diamond$ to it as unary operators that behave just the same way $\lnot$ does.
By this we mean in particular that $\Box (\varphi [t / x])$ (i.e.\ first substituting $t$ and then applying $\Box$) and $(\Box \varphi) [t / x]$ (first applying $\Box$ and then substituting $t$) are the same formula, just the same way $\lnot (\varphi [t / x])$ and $(\lnot \varphi) [t / x]$ are.

Now, enter

\begin{definition}\label{def:Kripke.sheaf}
A bounded morphism $\pi : (D, R_D) \to (X, R_X)$ is called a \emph{Kripke sheaf} over $(X, R_X)$ if
\begin{enumerate}
\setcounter{enumi}{\value{equation}}
\item\label{itm:sheaf.Kripke.1}
$a R_D b \pi w$ and $a R_D b' \pi w$ imply $b = b'$.
That is, $(R_D \cmp \opprel{R_D}) \cap (\opprel{\pi} \cmp \pi) \subseteq 1_D$.
\setcounter{equation}{\value{enumi}}
\end{enumerate}
\end{definition}

We fix one such map and, as we did in \autoref{sec:quantification.classical}, regard it as a residence map from the individuals $D$ to the worlds $X$.
Then, for each $n \in \NN$, the set $D^n_X$ of $n$-tuples from the same world comes with the ``epistemic'' relation $R_{D^n_X}$ by \autoref{thm:Kr.topological} or by \eqref{itm:Kripke.limit.product} and \eqref{itm:Kripke.limit.equalizer}.
Categorically, this is to take the $n$-fold pullback of $(D, R_D)$ over $(X, R_X)$ in $\Kr$, or equivalently the $n$-fold product of $\pi$ in the slice category $\Kr / (X, R_X)$.

We interpret first-order modal logic with $\pi$ and other structure in $\Kr / (X, R_X)$.
The classical base of the logic is interpreted with the underlying, non-Kripke structure in $\Sets / X$, just as in \autoref{sec:quantification.classical}.
The new, modal part is then added to the base using the Kripke structure, as follows:
First we require that, for each $n$-ary function symbol $f$, its interpretation $\Scott{f} : (D^n_X, R_{D^n_X}) \to (D, R_D)$ be monotone, so that all interpretations $\Scott{\, \bar{y} \mid t \,}$ of terms are monotone---%
i.e., they must be arrows of $\Kr / (X, R_X)$.
Then we set
\begin{align}
\label{eq:FOML.semantics.modal}
\Scott{\, \bar{x} \mid \Box \varphi \,} & = \forall_\opprel{R_{D^n_X}} \Scott{\, \bar{x} \mid \varphi \,} , &
\Scott{\, \bar{x} \mid \Diamond \varphi \,} & = \exists_\opprel{R_{D^n_X}} \Scott{\, \bar{x} \mid \varphi \,} .
\end{align}
In this way, we adopt the following ideas for the semantics.
\begin{enumerate}
\setcounter{enumi}{\value{equation}}
\item\label{itm:FOML.semantics.idea.1}
We use a family of Kripke models $(D^n_X, R_{D^n_X}, \Scott{-})$, where each $(D^n_X, R_{D^n_X})$ is the $n$-fold product of $\pi : (D, R_D) \to (X, R_X)$ in $\Kr / (X, R_X)$.
\item\label{itm:FOML.semantics.idea.2}
Each dual $(\pw(D^n_X), \forall_\opprel{R_{D^n_X}}, \exists_\opprel{R_{D^n_X}})$ is a CABAO of $n$-ary properties governed by \eqref{eq:FOML.semantics.boolean} and \eqref{eq:FOML.semantics.modal}.
\item\label{itm:FOML.semantics.idea.3}
We interpret terms with arrows $f$ of $\Kr / (X, R_X)$ between products $(D^n_X, R_{D^n_X})$.
\item\label{itm:FOML.semantics.idea.4}
The CABAOs interact with one another via cross-context operations, which are interpreted, as in \eqref{eq:FOML.semantics.quantify}--\eqref{eq:FOML.semantics.substitution}, with $\exists_f \dashv f^{-1} \dashv \forall_f$ of arrows $f$ of $\Kr / (X, R_X)$.
\setcounter{equation}{\value{enumi}}
\end{enumerate}
So, let us enter

\begin{definition}\label{def:Kripke.sheaf.model}
By a \emph{Kripke-sheaf model} we mean a Kripke sheaf $\pi : (D, R_D) \to (X, R_X)$ paired with a family of maps $\Scott{-}$ that assigns
\begin{itemize}
\item
an arrow $\Scott{f} : D^n_X \to D$ of $\Kr / (X, R_X)$ to each $n$-ary function symbol $f$,
\item
$\Scott{\, \bar{y} \mid t \,} : D^n_X \to D$ to all terms-in-contexts $(\, \bar{y} \mid t \,)$ by $\Scott{\, \bar{x} \mid f \bar{x} \,} = \Scott{f}$, $\Scott{\, \bar{x} \mid x_i \,} :: \bar{a} \mapsto a_i$, and \eqref{eq:FOML.semantics.term},
\item
any subset $\Scott{F} \subseteq D^n_X$ to each $n$-ary relation symbol $F$, and
\item
$\Scott{\, \bar{x} \mid \varphi \,} \subseteq D^n_X$ to all formulas-in-contexts $(\, \bar{x} \mid \varphi \,)$ by $\Scott{\, \bar{x} \mid F \bar{x} \,} = \Scott{F}$, \eqref{eq:FOML.semantics.boolean}--\eqref{eq:FOML.semantics.vacuous} and \eqref{eq:FOML.semantics.modal}.
\end{itemize}
\end{definition}

\autoref{def:Kripke.sheaf.model} requires $\pi$ to be not just a monotone map but moreover a Kripke sheaf, whereas no Kripke sheaves are mentioned in the ideas \eqref{itm:FOML.semantics.idea.1}--\eqref{itm:FOML.semantics.idea.4}.
The requirement is needed, however, precisely in order for the interaction \eqref{itm:FOML.semantics.idea.4} to behave coherently.
Given any $\Scott{\sigma} \subseteq D$, observe that there are two ways to obtain $\Scott{\, y \mid \Box \sigma \,}$ by applying \eqref{eq:FOML.semantics.vacuous} and \eqref{eq:FOML.semantics.modal}, viz.\
\begin{align*}
\Scott{\, y \mid \Box \sigma \,}
& = \forall_\opprel{R_D} \Scott{\, y \mid \sigma \,}
  = \forall_\opprel{R_D} \cmp \pi^{-1} \Scott{\sigma} , &
\Scott{\, y \mid \Box \sigma \,}
& = \pi^{-1} \Scott{\Box \sigma}
  = \pi^{-1} \cmp \forall_\opprel{R_X} \Scott{\sigma} .
\end{align*}
So the well-definedness of $\Scott{-}$, along with \eqref{eq:FOML.semantics.vacuous} and \eqref{eq:FOML.semantics.modal}, requires that $\forall_\opprel{R_D} \cmp \pi^{-1} = \pi^{-1} \cmp \forall_\opprel{R_X}$, or dually $\pi \cmp R_D = R_X \cmp \pi$, i.e.\ that $\pi$ be a bounded morphism.
Indeed, any map $\Scott{\, \bar{y} \mid \bar{t} \,}$ involved in \eqref{eq:FOML.semantics.substitution} must be a bounded morphism.
Recall that our syntax has $\Box (\varphi [\bar{t} / \bar{x}]) = (\Box \varphi) [\bar{t} / \bar{x}]$.
This means that, for $\Scott{-}$ to be well-defined, we need $\Scott{\, \bar{y} \mid \Box (\varphi [\bar{t} / \bar{x}]) \,} = \Scott{\, \bar{y} \mid (\Box \varphi) [\bar{t} / \bar{x}] \,}$, both sides giving the same interpretation to the same formula $\Box \varphi [\bar{t} / \bar{x}]$.
So, given $\Scott{\, \bar{x} \mid \varphi \,} \subseteq D^n_X$ and
$\Scott{\bar{t}} = \Scott{\, \bar{y} \mid \bar{t} \,} : D^m_X \to D^n_X$,
\eqref{eq:FOML.semantics.substitution} and \eqref{eq:FOML.semantics.modal} imply
\begin{gather*}
\forall_\opprel{R_{D^m_X}} \cmp \Scott{\bar{t}}^{-1} \Scott{\, \bar{x} \mid \varphi \,}
= \Scott{\, \bar{y} \mid \Box (\varphi [\bar{t} / \bar{x}]) \,}
= \Scott{\, \bar{y} \mid (\Box \varphi) [\bar{t} / \bar{x}] \,}
= \Scott{\bar{t}}^{-1} \cmp \forall_\opprel{R_{D^n_X}} \Scott{\, \bar{x} \mid \varphi \,} .
\end{gather*}
Thus, the well-definedness of $\Scott{-}$, along with \eqref{eq:FOML.semantics.substitution} and \eqref{eq:FOML.semantics.modal}, again requires that $\Scott{\bar{t}}$ be a bounded morphism.%
\footnote{\strut%
From a perspective of categorical logic, one often takes \eqref{eq:FOML.semantics.substitution}, for all $\varphi$, as part of the definition of a model, rather than a derived fact about the model.
It is from this perspective that we describe the situation as a matter of well-definedness of the model.
One could also see the same situation as a matter of deriving \eqref{eq:FOML.semantics.substitution} from its atomic case using a property of bounded morphisms;
we acknowledge an anonymous reviewer for this perspective.
One could of course choose to reject \eqref{eq:FOML.semantics.substitution} or \eqref{eq:FOML.semantics.modal}, or even to use a syntax without $\Box (\varphi [t / x]) = (\Box \varphi) [t / x]$.
(These options, needless to say, would make \autoref{thm:FOML.completeness} unavailable to one's semantics.)
A notable case of rejecting \eqref{eq:FOML.semantics.modal} is the counterpart theory in \cite{lew68}, which restricts \eqref{eq:FOML.semantics.modal} to the case where all the variables in $\bar{x}$ actually occur freely in $\varphi$.
}
Yet, all maps involved in \autoref{def:Kripke.sheaf.model} are indeed guaranteed to be bounded morphisms, by

\begin{fact}\label{thm:sheaves.full.subcat}
If $\pi : D \to X$ is a Kripke sheaf, then so is every $\pi^n : D^n_X \to X$.
Moreover, given two Kripke sheaves $\pi_D : D \to X$ and $\pi_E : E \to X$, any monotone map $f : D \to E$ over $X$ (i.e.\ satisfying $\pi_E \cmp f = \pi_D$) is also a Kripke sheaf (and hence a bounded morphism).
On the other hand, $\pi$ is a Kripke sheaf iff both $\pi$ and the ``diagonal map'' $\Delta : D \to D^2_X :: a \mapsto (a, a)$ are bounded morphisms.%
\footnote{\strut%
See Facts 4.2, 4.4, and 4.6 in \cite{kis11}.
}
\end{fact}

In short, the simple combination of \eqref{eq:FOML.semantics.term}--\eqref{eq:FOML.semantics.vacuous}, for classical first-order logic, and \eqref{eq:FOML.semantics.boolean} and \eqref{eq:FOML.semantics.modal}, for propositional modal logic, is made possible by Kripke sheaves and \autoref{thm:sheaves.full.subcat}.
And this simple combination makes the logic of Kripke-sheaf semantics the simple union of classical first-order logic and modal logic.

\begin{fact}\label{thm:FOML.completeness}
Let $\sys{FOK}$ be the first-order modal logic that consists of all the rules and axioms of classical first-order logic, and the rules and axioms of propositional modal logic $\sys{K}$.
Then $\sys{FOK}$ is sound and complete with respect to the Kripke-sheaf models.
The same holds with $\sys{S4}$ (or $\sys{S5}$, respectively) in place of $\sys{K}$, with respect to the subclass of Kripke-sheaf models over preorders (or equivalence relations).%
\footnote{\strut%
See, e.g., Corollary 6.1.24 of \cite{gab09}.
}
\end{fact}

\subsection{First-Order Dynamic Epistemic Logic}\label{sec:quantification.fodel}

In \autoref{sec:quantification.kripke} we saw how the Kripke-sheaf structure extended the modal logic of a Kripke model to the first order.
We will now lay out how the same structure can extend the product update of Kripke models to the first order.%
\footnote{\strut%
A sheaf semantics for first-order PAL was given (in a more general, neighborhood setting) in \cite{kis13}.
A first-order extension of PAL was also given briefly in \cite{ma11}, which, however, used constant domains to interpret quantification.
See footnote \ref{fn:constant.domain} as well.
}
One remark is in order:
We saw in Sections \ref{sec:Kripke} and \ref{sec:del} that, whereas bounded morphisms play a more prominent r\^ole than merely monotone maps in the semantics of static modal logic, merely monotone maps are essential in the semantics of DEL\@.
This theme recurs in this subsection.
In \autoref{sec:quantification.kripke}, we reviewed the fact that static first-order modal logic needed Kripke sheaves to make sure all the maps involved were bounded morphisms.
In our new semantics for first-order DEL, however, the structure of the category $\Kr$ of monotone maps will play a central r\^ole again.

Let $(\pi : (D, R_D) \to (X, R_X), \Scott{-}_\pi)$ be a Kripke-sheaf model, and $(E, R_E)$ be a Kripke frame.
We regard the latter as an event model, and assume that preconditions $\Pre(e)$ for $e \in E$ are all (closed) sentences, so that $\Scott{\Pre(e)}_\pi$ makes sense and $\Scott{\Pre(e)}_\pi \subseteq X$.
Then $(X, R_X)$ is product-updated with $(E, R_E)$ into $(X \otimes E, R_{X \otimes E})$.
For the first-order structure, we moreover ``pullback-update'' $(\pi, \Scott{-}_\pi)$, by pulling everything back along the projection $p_X : X \otimes E \to X$.
Recall that $(\pi, \Scott{-}_\pi)$ uses the structure of the slice category $\Kr / (X, R_X)$;
hence $p_X$ induces a pullback functor ${p_X}^\ast : \Kr / (X, R_X) \to \Kr / (X \otimes E, R_{X \otimes E})$.
So we apply this to obtain an updated residence map $\pi_{X \otimes E} = {p_X}^\ast \pi : D_{X \otimes E} \to X \otimes E$, and to obtain $\Scott{-}_{\pi_{X \otimes E}}$ from $\Scott{f}_{\pi_{X \otimes E}} = {p_X}^\ast \Scott{f}_\pi$ for function symbols $f$ and $\Scott{F}_{\pi_{X \otimes E}} = {p_X}^\ast \Scott{F}_\pi$ for relation symbols $F$.
We need to note that the structure of $\Kr$ is essential for the pullback update.
Pullbacks are taken in the category $\Kr$ of monotone maps in general as opposed to bounded morphisms, and along the map $p_X : X \otimes E \to X$, which, as seen in \autoref{sec:del.del}, must not be a bounded morphism for DEL to show interesting behaviors.

Here is an explicit description of the pullback update:
\begin{itemize}
\item
Using the notation $D_w = \pi^{-1}(\{ w \})$, the pullback of $\pi^n : D^n_X \to X :: \bar{a} \mapsto \pi(a_i)$ along $p_X : X \otimes E \to X :: (w, e) \mapsto w$ has the domain
\begin{gather*}
\sum_{(w, e) \in X \otimes E} D_w^n = \{\, (\bar{a}, e) \in D^n_X \times E \mid \pi^n(\bar{a}) \in \Scott{\Pre(e)}_\pi \,\} ,
\end{gather*}
for which we write $D^n_{X \otimes E}$, and projections
\begin{align*}
\pi^n_{X \otimes E} = {p_X}^\ast \pi^n & : D^n_{X \otimes E} \to X \otimes E :: (\bar{a}, e) \mapsto (\pi^n(\bar{a}), e) , \\
p_{D^n_X} & : D^n_{X \otimes E} \to D^n_X :: (\bar{a}, e) \mapsto \bar{a} .
\end{align*}
It also comes with another projection $p_{E, n} : D^n_{X \otimes E} \to E :: (\bar{a}, e) \mapsto e$.
The ``epistemic'' relation $R_{D^n_{X \otimes E}}$ on $D^n_{X \otimes E}$ is an initial lift, viz.\
\begin{gather*}
R_{D^n_{X \otimes E}} = (\opprel{p_{D^n_X}} \cmp R_{D^n_X} \cmp p_{D^n_X}) \cap (\opprel{p_{E, n}} \cmp R_E \cmp p_{E, n}) , \\
\text{i.e.}\quad
(\bar{a}, e_1) R_{D^n_{X \otimes E}} (\bar{b}, e_2) \iff a_1 R_D b_1 \text{, \ldots, } a_n R_D b_n \text{ and } e_1 R_E e_2 .
\end{gather*}
\item
For an $n$-ary function symbol $f$, we have $\Scott{f}_\pi : D^n_X \to D$ and then
\begin{gather*}
\Scott{f}_{\pi_{X \otimes E}} = {p_X}^\ast \Scott{f}_\pi : D^n_{X \otimes E} \to D_{X \otimes E} :: (\bar{a}, e) \mapsto (\Scott{f}_\pi(\bar{a}), e) .
\end{gather*}
\item
For an $n$-ary relation symbol $F$, we have $\Scott{F}_\pi \subseteq D^n_X$ and then
\begin{align*}
\Scott{F}_{\pi_{X \otimes E}}
& = {p_X}^\ast \Scott{F}_\pi
  = {p_{D^n_X}}^{-1} \Scott{F}_\pi
  = \{\, (\bar{a}, e) \in D^n_{X \otimes E} \mid \bar{a} \in \Scott{F}_\pi \,\} \subseteq D^n_{X \otimes E} .
\end{align*}
\end{itemize}

The pullback update indeed updates a Kripke-sheaf model to another:

\begin{theorem}\label{thm:pullback.update}
Given a Kripke-sheaf model $(\pi, \Scott{-}_\pi)$, its pullback update $(\pi_{X \otimes E}, \Scott{-}_{\pi_{X \otimes E}})$ along $p_X : X \otimes E \to X$ is a Kripke-sheaf model.
\end{theorem}

\begin{proof}
As in \autoref{thm:sheaves.full.subcat}, both $\pi$ and the diagonal map $\Delta$ of $\pi$ are bounded morphisms, and hence \autoref{thm:pullback.preserve.open} implies that both $\pi_{X \otimes E} = {p_X}^\ast \pi$ and ${p_X}^\ast \Delta$ are bounded morphisms.
Yet ${p_X}^\ast \Delta$ is the diagonal map of $\pi_{X \otimes E}$, since the pullback functor ${p_X}^\ast$ preserves finite limits.
Therefore $\pi_{X \otimes E}$ is a Kripke sheaf by \autoref{thm:sheaves.full.subcat}.
Moreover, for each $n \in \NN$, ${p_X}^\ast \pi^n$ is the $n$-fold product of $\pi_{X \otimes E}$ over $X \otimes E$, since ${p_X}^\ast$ preserves finite limits.
\end{proof}

Now we have two Kripke-sheaf models, $(\pi, \Scott{-}_\pi)$ before update and $(\pi_{X \otimes E}, \Scott{-}_{\pi_{X \otimes E}})$ after, and we can use relations between them to interpret the DEL operators $\nec{E, e}$ and $\pos{E, e}$.
Here is a key idea:
As in \eqref{itm:FOML.semantics.idea.1}--\eqref{itm:FOML.semantics.idea.2}, each sheaf model has a Kripke model for $n$-ary properties, $D^n_X$ and $D^n_{X \otimes E}$;
so we treat $D^n_X$ and $D^n_{X \otimes E}$ as the product-update structure of \autoref{sec:del.del} that interprets the application of $\nec{E, e}$ and $\pos{E, e}$ to $n$-ary formulas-in-contexts.
Since $\Scott{\, \bar{x} \mid \Pre(e) \,}_\pi = (\pi^n)^{-1} \Scott{\Pre(e)}$, observe
\begin{gather*}
D^n_{X \otimes E} = \sum_{e \in E} \Scott{\, \bar{x} \mid \Pre(e) \,}_\pi = \{\, (\bar{a}, e) \in D^n_X \times E \mid \bar{a} \in \Scott{\, \bar{x} \mid \Pre(e) \,}_\pi \,\}
\end{gather*}
and note the similarity to \eqref{eq:del.product.update}.
We moreover have canonical maps as with \eqref{eq:del.product.update}, viz.\ the projection $p_{D^n_X} : D^n_{X \otimes E} \to D^n_X$ above and, for each $e \in E$,
\begin{itemize}
\item
The inclusion map $i^n_e : \Scott{\, \bar{x} \mid \Pre(e) \,}_\pi \incto D^n_X$.
\item
The coproduct injection $q^n_e : \Scott{\, \bar{x} \mid \Pre(e) \,}_\pi \to D^n_{X \otimes E} :: \bar{a} \mapsto (\bar{a}, e)$.
\end{itemize}
These maps tabulate a relation, $R^n_e = q^n_e \cmp \opprel{i^n_e} : D^n_X \relto D^n_{X \otimes E}$, which is dual to the two maps
\begin{align*}
\forall_\opprel{R^n_e} = \forall_{i^n_e} \cmp (q^n_e)^{-1},
\exists_\opprel{R^n_e} & = \exists_{i^n_e} \cmp (q^n_e)^{-1} : \pw(D^n_{X \otimes E}) \to \pw(D^n_X) .
\end{align*}
These then interpret $\nec{E, e}$ and $\pos{E, e}$ applied to $n$-ary formulas-in-contexts $(\, \bar{x} \mid \varphi \,)$, i.e.,
\begin{align*}
\Scott{\, \bar{x} \mid \nec{E, e} \varphi \,}_\pi & = \forall_\opprel{R^n_e} \Scott{\, \bar{x} \mid \varphi \,}_{\pi_{X \otimes E}} , &
\Scott{\, \bar{x} \mid \pos{E, e} \varphi \,}_\pi & = \exists_\opprel{R^n_e} \Scott{\, \bar{x} \mid \varphi \,}_{\pi_{X \otimes E}} ,
\end{align*}
which is just an ``in context'' version of \eqref{itm:DEL.semantics.modal}.

This defines our sheaf semantics for first-order DEL---%
but we need to check its well-definedness, similarly to the remark following \autoref{def:Kripke.sheaf.model}.
That is, we need
\begin{align*}
\forall_\opprel{R^m_e} \cmp {\Scott{\bar{t}}_{\pi_{X \otimes E}}}^{-1} \Scott{\, \bar{x} \mid \varphi \,}_{\pi_{X \otimes E}}
& = \Scott{\, \bar{y} \mid \nec{E, e} (\varphi [\bar{t} / \bar{x}]) \,}_\pi \\
& = \Scott{\, \bar{y} \mid (\nec{E, e} \varphi) [\bar{t} / \bar{x}] \,}_\pi
  = {\Scott{\bar{t}}_\pi}^{-1} \cmp \forall_\opprel{R^n_e} \Scott{\, \bar{x} \mid \varphi \,}_{\pi_{X \otimes E}} ,
\end{align*}
and similarly for $\pos{E, e}$.
Yet these are the case because $\Scott{\bar{t}}_{\pi_{X \otimes E}} \cmp R^m_e = R^n_e \cmp \Scott{\bar{t}}_\pi$ by

\begin{theorem}\label{thm:pullback.update.commute}
For any arrow $f : D^m_X \to D^n_X$ of $\Kr / (X, R_X)$,
\begin{align*}
{p_X}^\ast f \cmp R^m_e & = R^n_e \cmp f , &
\opprel{({p_X}^\ast f)} \cmp R^n_e & = R^m_e \cmp \opprel{f} .
\end{align*}
\end{theorem}

\begin{proof}
This follows from \autoref{thm:beck.chevalley} since the following squares are both pullbacks in $\Sets$.
\begin{gather*}
\begin{gathered}
\begin{tikzpicture}[x=20pt,y=20pt]
\coordinate (O) at (0,0);
\coordinate (r) at (5,0);
\coordinate (d) at (0,-2.5);
\node (A0) [inner sep=0.25em] at (O) {$D^m_{X \otimes E}$};
\node (A1) [inner sep=0.25em] at ($ (A0) + (r) $) {$(\pi^m)^{-1} \Scott{\Pre(e)}_\pi$};
\node (A2) [inner sep=0.25em] at ($ (A1) + (r) $) {$D^m_X$};
\node (B0) [inner sep=0.25em] at ($ (A0) + (d) $) {$D^n_{X \otimes E}$};
\node (B1) [inner sep=0.25em] at ($ (B0) + (r) $) {$(\pi^n)^{-1} \Scott{\Pre(e)}_\pi$};
\node (B2) [inner sep=0.25em] at ($ (B1) + (r) $) {$D^n_X$};
\draw [>->] (A1) -- (A0) node [pos=0.5,inner sep=2pt,above] {$q^m_e$};
\draw [right hook->] (A1) -- (A2) node [pos=0.5,inner sep=2pt,above] {$i^m_e$};
\draw [>->] (B1) -- (B0) node [pos=0.5,inner sep=2pt,below] {$q^n_e$};
\draw [right hook->] (B1) -- (B2) node [pos=0.5,inner sep=2pt,below] {$i^n_e$};
\draw [->] (A0) -- (B0) node [pos=0.5,inner sep=2pt,left] {${p_X}^\ast f$};
\draw [->] (A1) -- (B1);
\draw [->] (A2) -- (B2) node [pos=0.5,inner sep=2pt,right] {$f$};
\coordinate (A1-pb-l) at ($ (A1) + (-1,-1) $);
\coordinate (A1-pb-r) at ($ (A1) + (1,-1) $);
\draw ($ (A1-pb-l) + (0.45,0) $) -- (A1-pb-l) -- ($ (A1-pb-l) + (0,0.45) $);
\draw ($ (A1-pb-r) + (-0.45,0) $) -- (A1-pb-r) -- ($ (A1-pb-r) + (0,0.45) $);
\end{tikzpicture}
\end{gathered}
\qedhere
\end{gather*}
\end{proof}

Now, the semantics validates all the reduction axioms of propositional DEL, simply because $D^n_{X \times E}$ is just the product update of $D^n_X$ with $(E, R_E)$.
One more reduction axiom is needed, however---%
viz.\ for quantifiers.
And here it is:
\begin{gather}
\label{eq:FODEL.reduction.exists}
\nec{E, e} \forall y \ldot \varphi \equiv \forall y \ldot \nec{E, e} \varphi .
\end{gather}

\begin{proof}
We show the validity of \eqref{eq:FODEL.reduction.exists}.
Let $p : D^{n + 1}_X \to D^n_X :: (\bar{a}, b) \mapsto \bar{a}$.
Then \autoref{thm:pullback.update.commute} implies $\opprel{({p_X}^\ast p)} \cmp R^n_e = R^{n + 1}_e \cmp \opprel{p}$, and dually $\forall_\opprel{R^n_e} \cmp \forall_{{p_X}^\ast p} = \forall_p \cmp \forall_\opprel{R^{n + 1}_e}$.
Therefore
\begin{align*}
\Scott{\, \bar{x} \mid \nec{E, e} \forall y \ldot \varphi \,}_\pi
& = \forall_\opprel{R^n_e} \cmp \forall_{{p_X}^\ast p} \Scott{\, \bar{x}, y \mid \varphi \,}_{\pi_{X \otimes E}}
  = \forall_p \cmp \forall_\opprel{R^{n + 1}_e} \Scott{\, \bar{x}, y \mid \varphi \,}_{\pi_{X \otimes E}}
  = \Scott{\, \bar{x} \mid \forall y \ldot \nec{E, e} \varphi \,}_\pi .
\qedhere
\end{align*}
\end{proof}

This now gives a completeness result extending \autoref{thm:FOML.completeness} by the standard method of reduction.

\begin{theorem}
Let $\sys{FODEL\text{-}K}$ be the first-order modal logic that consists of $\sys{FOK}$, all the reduction axioms of propositional DEL, and \eqref{eq:FODEL.reduction.exists}.
Then $\sys{FODEL\text{-}K}$ is sound and complete with respect to the Kripke-sheaf models with pullback updates.
The versions with $\sys{S4}$ and $\sys{S5}$ in place of $\sys{K}$ hold with respect to the obvious subclasses of Kripke-sheaf models.
\end{theorem}

\section{Connections to Preceding Approaches}\label{sec:connections}

There have been approaches to modal logic and DEL that take advantage of concepts and methods of category theory in different ways from our approach.
This section discusses connections between some of these approaches and ours.%
\footnote{\strut%
We thank anonymous reviewers for references, and for their suggestions that the connections should be discussed.
}

Semantics of modal logic $\sys{S4}$ shows various categorical structures.
A Kripke frame $(X, \precsim)$ for $\sys{S4}$ is a preorder, and hence itself a category.
Also, the family $\Op X$ of ${\precsim}$-upward closed subsets of $X$ forms a topology on $X$, and hence a category.
Moreover, the interior operation $\int : \pw X \to \Op X$ of this topology is right adjoint to the inclusion $i : \Op X \incto \pw X$, so that $\Box = i \cmp \int$ is the comonad of the adjunction.%
\footnote{\strut%
See Section 10.4 of \cite{awo10} and Subsection 5.1.1 of \cite{jac16} for comonads.
In fact, instead of a poset $\pw X$ one can take a general category $\C$ and a comonad $\Box$ on $\C$ to interpret $\sys{S4}$ (perhaps with a non-modal base weaker than classical);
see e.g.\ \cite{ale01}.
}
The notion of (Kripke) sheaf lifts all these structures to the first order:
A Kripke sheaf over a preorder $(X, {\precsim})$ is equivalently a ``presheaf'' on the category $(X, {\precsim})$, an ``\'etale space'' over the space $(X, \Op X)$, and a ``sheaf'' on the category $\Op X$.%
\footnote{\strut%
See Chapters I through III of \cite{mac92} for these concepts.
}
Moreover, the adjunction $i \dashv \int$ is lifted to a ``geometric morphism'' $i^\ast \dashv i_\ast$ from $\Sets / X$ to the ``topos'' of sheaves over $\Op X$, so that its comonad $i^\ast \cmp i_\ast$ induces $\Box : \pw D \to \pw D$ for every Kripke sheaf $\pi : (D, {\precsim}_D) \to (X, {\precsim})$.%
\footnote{\strut%
See Chapter VII of \cite{mac92} for geometric morphisms in general, and Section 5.2 of \cite{awo08} for the geometric-morphism interpretation of $\sys{FOS4}$.
}
Not all these categorical structures carry over to the general (i.e.\ non-$\sys{S4}$) Kripke semantics.
It will be interesting, however, to investigate how to integrate them with DEL updates, given that epistemic relations are normally assumed to be preorders.
In fact, given a monotone map $f : (X, {\precsim}_X) \to (Y, {\precsim}_Y)$ of preorders, the pullback functor $f^\ast$ (which plays a key r\^ole in the pullback update of \autoref{sec:quantification.fodel}) has a right adjoint $f_\ast$, and $f^\ast \dashv f_\ast$ is a typical example of geometric morphism, from the topos of Kripke sheaves over $(X, {\precsim}_X)$ to those over $(Y, {\precsim}_Y)$.

A categorical approach that covers the entire Kripke semantics (for static modal logic) is given by coalgebras (see, e.g., \cite{mos99,cir11,kup11,jac16}).
The category $\Rel$ of relations is the ``Kleisli category'' of the ``powerset monad'' $\pw : \Sets \to \Sets$, meaning, among other things, that the relations $R : X \relto Y$ correspond 1--1 to the functions $r : X \to \pw Y$.%
\footnote{\strut%
This correspondence can also be described as between $R : X \times Y \to \2$ and $r : X \to (Y \to \2)$.
See Chapter VI of \cite{mac98}, Chapter 10 of \cite{awo10}, and Chapter 5 of \cite{jac16} for monads and their Kleisli categories, and $\pw$ and $\Rel$ as an example.
\autoref{thm:equivalence.Kripke} can then be read as stating that $\exists_-$ is a ``comparison functor'' that presents $\Rel$ as the category $\CABAvee$ of free algebras of $\pw$.
}
Indeed, the powerset monad is precisely the duality $\exists_- : \Rel \to \CABAvee$ restricted to $\Sets$ (and followed by the forgetful $U : \CABAvee \to \Sets$).
The correspondence implies that the Kripke frames $(X, R : X \relto X)$ are exactly the coalgebras $r : X \to \pw X$ for the endofunctor $\pw$.
Their homomorphisms, from $r_X : X \to \pw X$ to $r_Y : Y \to \pw Y$, are normally defined as functions $f : X \to Y$ satisfying $\exists_f \cmp r_X = r_Y \cmp f$, which amounts to \eqref{itm:open.Kripke.2}, $f \cmp R_X = R_Y \cmp f$, for the corresponding relations $R_X$ and $R_Y$.
Therefore, in the coalgebraic approach to Kripke semantics, $\Coalg(\pw)$, the category of coalgebras and their homomorphisms normally considered, is---%
like the category $\CABAO$ of CABAOs and their homomorphisms---%
equivalent to the category $\Krb$ of bounded morphisms.
In this article, on the other hand, we emphasized the significance of the topological category $\Kr$ of monotone maps for DEL\@.%
\footnote{\strut%
One can of course express $\Kr$ with coalgebras, by defining a weaker notion of homomorphism, corresponding to monotone maps---%
i.e., a function $f : X \to Y$ is ``continuous'' from $r_X : X \to \pw X$ to $r_Y : Y \to \pw Y$ if $\exists_f \cmp r_X \leqslant r_Y \cmp f$ (i.e.\ $\exists_f \cmp r_X(x) \subseteq r_Y \cmp f(x)$ for all $x \in X$).
On the other hand, Kripke sheaves can be defined within $\Coalg(\pw)$.
One can rewrite \eqref{itm:sheaf.Kripke.1} as a homomorphism $\pi : D \to X$ from $r_D : D \to \pw D$ to $r_X : X \to \pw X$ satisfying
\begin{itemize}
\item
for each $a \in D$, the restriction of $\pi$ to $r_D(a)$ is an injection.
\end{itemize}
Or it may be better to use the characterization in \autoref{thm:sheaves.full.subcat}---%
i.e., $\pi$ is a Kripke sheaf iff both $\pi$ and $\Delta$ are homomorphisms.
See Fact 4.2 of \cite{kis11}.
The latter definition can indeed be extended to more kinds of coalgebras and not just Kripke frames.
}

There have in fact been algebraic \cite{kur13} and coalgebraic \cite{bal03,cir07} approaches to DEL\@.
In particular, the algebraic approach by Kurz and Palmigiano \cite{kur13} uses ideas closely related to those in \autoref{sec:del} of this article:
They observe that the product update $X \otimes E$ is a subframe of the coproduct $X \times E = \sum_{e \in E} X$, and study the dual structure, i.e.\ a quotient of the product $\prod_{e \in E} \pw(X)$.%
\footnote{\strut%
It is therefore the maps $i$ and $q'_e$ in \eqref{eq:del.product.update.diagram} that play a central r\^ole in \cite{kur13}.
In contrast, we put more emphasis on $q_e$ and $i_e$, though $R_e = q_e \cmp \opprel{i_e} = \opprel{i} \cmp q'_e$ as noted on p.\ \pageref{page:beck.chevalley.product.update}.
Also, in our treatment, the characterization of $X \times E$ as a product plays a key role as well, since the Kripke frame on $X \times E$ is the product of $X$ and $E$, but not the coproduct of $(X)_{e \in E}$, in $\Kr$.
Moreover, we treat $\Pre(e) \mimp {-}$ and $\Pre(e) \wedge {-}$ as the modal operators of $i_e \cmp \opprel{i_e}$, a perspective that then enables us to prove the reduction axioms \eqref{itm:PAL.reduction.box} and \eqref{itm:DEL.reduction.box} directly by the relation-modality duality $\forall_\opprel{-}$.
This should be contrasted to the treatment of $\Pre(e) \mimp {-}$ and $\Pre(e) \wedge {-}$ in proofs in Section 7 of \cite{kur13}.
}
Kurz and Palmigiano are well aware that these constructions do not take place in $\Krb$ or $\CABAO$ but rather in $\Kr$ and $\CABAOc$.
They stop short, however, of studying $\Kr$ or $\CABAOc$, saying that
``for these dual characterizations to be defined, an \textit{a priori} specification of the fully fledged category-theoretic environment in which these constructions are taken is actually not needed'' (\cite{kur13}, 2).
We, in contrast, work under the philosophy that, when one finds a good heuritstics that leads to a new result, they should study the heuritstics itself and shape it into a theory that yields more results systematically.
The point of \autoref{sec:quantification} was to demonstrate how to put to use more structures in $\Kr$.
It should also be stressed that we use one more category, viz.\ $\Rel$, and take essential advantage of the fundamental relation-modality dualities of \autoref{sec:Kripke.duality}, and not just the derivative dualities of \autoref{sec:Kripke.frames} between Kripke frames and CABAOs.

\section{Conclusion and Future Work}\label{sec:conclusion}

In this paper we have recast the standard semantics of dynamic epistemic logic (DEL) in categorical terms and shed new structural light on it.
It should be clear by now how conceptually powerful the new way of applying categorical method is:
As demonstrated by our new semantics for first-order DEL, our categorical, structural perspective tightly connects what we want (or need) logically or syntactically and what we need (or want) semantically.

Our new application of the categorical methodology promises to be helpful on multiple fronts of the study of DEL\@.
Naturally expected future work is to extend our approach to more vocabulary (e.g.\ common knowledge or $\mu$-calculus), more types of logic (e.g.\ higher-order DEL or typed DEL), more structures (e.g.\ probability), and more general settings (e.g.\ intuitionistic or constructive modal logic).
Various updates can be expressed as functors between categories of models, and these expressions are expected to help characterize properties of updates such as the preservation of constructions or the admitting of reduction axioms.
As mentioned in \autoref{sec:connections}, the case of $\sys{S4}$ can be formulated in terms of toposes.
Or our structural, topological ideas on the category $\Kr$ of monotone maps for DEL can be used to augment the coalgebraic generalization of the subcategory $\Krb$ of bounded morphisms.
One may also find, e.g., \eqref{eq:FODEL.reduction.exists} too strong for their purpose, and hence need to replace the pullback update with a more flexible idea.
Furthermore, although we formulated a categorical semantics, we did not mention a crucial aspect of categorical logic---%
viz.\ an interpretation $\Scott{-}$ as a homomorphism.
To cover this aspect we need to define a ``syntactic category'' for DEL;
this will then lead to a new theory of duality.

\nocite{*}
\bibliographystyle{eptcs}
\bibliography{del.cat}
\end{document}